\documentclass[prb,aps,10pt,english,twocolumn,amsmath,amssymb,superscriptaddress]{revtex4-2}
\usepackage{amsmath}
\usepackage{graphicx}
\usepackage{amssymb}
\usepackage{dsfont}
\usepackage{color}
\usepackage[bookmarks=true,colorlinks,linkcolor=blue,urlcolor=blue,citecolor=blue]{hyperref}
\usepackage{braket}
\usepackage{amsthm}
\usepackage[normalem]{ulem}
  
\newcommand{\new}[1]{\textcolor{black}{{#1}}}

\usepackage{hyperref}
\usepackage{tikz}
\usetikzlibrary{calc}

\newcommand{\be}{\begin{equation}}
\newcommand{\ee}{\end{equation}}
\newcommand{\bea}{\begin{eqnarray}}
\newcommand{\eea}{\end{eqnarray}}
\newcommand{\mc}{\mathcal}
\newcommand{\mb}{\mathbf}

\newtheorem{lemma}{Lemma}

\newtheorem{theorem}{Theorem}
\newtheorem{proposition}{Proposition}

\begin{document}

\title{Partons from stabilizer codes}

\author{Rafael A. Mac\^{e}do}
\affiliation{Departamento de F\'{i}sica Te\'{o}rica e Experimental, Universidade Federal do Rio Grande do Norte, Natal, RN, 59078-970, Brazil}
\affiliation{International Institute of Physics, Universidade Federal do Rio Grande do Norte, Natal, RN, 59078-970, Brazil}
\author{Carlo C. Bellinati}
\affiliation{
Instituto de F\'isica, Universidade de S\~ao Paulo, 05315-970 S\~ao Paulo, SP, Brazil}
\author{Weslei B. Fontana}
\affiliation{Department of Physics, National Tsing Hua University, Hsinchu 30013, Taiwan}
\author{Eric C. Andrade}
\affiliation{
Instituto de F\'isica, Universidade de S\~ao Paulo, 05315-970 S\~ao Paulo, SP, Brazil}
\author{Rodrigo G. Pereira}
\affiliation{Departamento de F\'{i}sica Te\'{o}rica e Experimental, Universidade Federal do Rio Grande do Norte, Natal, RN, 59078-970, Brazil}
\affiliation{International Institute of Physics, Universidade Federal do Rio Grande do Norte, Natal, RN, 59078-970, Brazil}

\begin{abstract}
The Gutzwiller projection of fermionic wave functions is a well-established method for generating variational wave functions describing exotic states of matter, such as quantum spin liquids. We investigate the conditions under which a projected wave function constructed from fermionic partons  can be rigorously shown to possess topological order. We demonstrate that these conditions can be precisely determined  in the case of projected Majorana stabilizer codes. We then use matrix product states to study  states that interpolate between two distinct Majorana fermion codes, one yielding a  $\mathbb Z_2$ spin liquid and the other a trivial polarized state upon projection. While the free-fermion states are adiabatically connected,  we find that the projected states undergo  a  phase transition detected by the topological entanglement entropy. Our work underscores the profound impact of the Gutzwiller projection and cautions against inferring properties of quantum spin liquids solely from their unprojected counterparts.

\end{abstract}
\maketitle

\section{Introduction}

Topologically ordered states are fascinating examples of how quantum phases of matter find  applications in quantum information processing \cite{zeng2019quantum}. Originally defined in terms of a robust  ground state degeneracy that depends on the genus of a compact two-dimensional   surface \cite{Wen1990}, topological order is also associated with fractional excitations with anyonic statistics \cite{Oshikawa2006} and long-range entanglement detected by the topological entanglement entropy \cite{kitaev2006topological,levin2006detecting}. The nonlocal properties of anyons can be exploited  in quantum error correcting codes  \cite{kitaev2003fault} and have been demonstrated  in recent experiments in quantum simulation platforms  \cite{Satzinger2021,Semeghini2021,IqbalCP2024}.

Topological order is a hallmark of gapped quantum spin liquids (QSLs) \cite{savary2016quantum,ZhouRMP2017,Broholm2020}.  In particular, the paradigmatic example of a QSL, Anderson's resonating  valence bond (RVB) state \cite{Anderson1987}, exhibits  two types of excitations with nontrivial mutual statistics that correspond to the anyons of $\mathbb Z_2$ topological order \cite{Kivelson1987,sachdev2023quantum}. The standard analytical approach to describe QSLs is to formulate a mean-field theory \cite{Baskaran1987,WenPRB1991} in which the local spin degrees of freedom are fractionalized into  quasiparticles called partons \cite{savary2016quantum,sachdev2023quantum}.  This type of approximation  neglects the   fluctuations of an emergent gauge field that mediates interactions between the partons, but  it correctly captures some universal   properties  if the effective gauge theory is in a deconfined phase.

The parton mean-field approach  provides a way to construct wave functions for QSLs. Starting from the ground state of a trial mean-field Hamiltonian for fermionic partons, one obtains a spin wave function by applying a Gutzwiller projection that imposes a local single-occupancy constraint  \cite{GROS198953}. From the projected fermionic state, physical observables can then be calculated  numerically using    variational Monte Carlo \cite{GROS198953,becca2017quantum}. A more recent alternative employs matrix product states (MPS) to compute a compressed representation of the projected state \cite{Jin2020,wu2020tensor,petrica2021finite}.

The usual expectation is that the unprojected fermionic state and the spectrum of the mean-field Hamiltonian give insight into  the properties of the physical spin state. For instance, the Gutzwiller projection of a BCS wave function generates a wave function for the short-range RVB state \cite{Anderson1987}, and the superconducting pairing in the mean-field ansatz is interpreted in terms of  the Anderson-Higgs mechanism that gaps out the gauge field  in the  $\mathbb Z_2$ QSL \cite{WenPRB1991}.  However,  a superconducting wave function does not necessarily lead to a topologically ordered state upon projection \cite{IvanovPRB2002,Paramekanti2005,Li2007}. Moreover, selecting a projected wave function by  minimizing the energy of a given spin Hamiltonian may fail to reveal topological order, as it is possible to construct a topologically trivial variational state that provides an arbitrarily accurate approximation to the ground-state energy density in a topological phase \cite{Balents2014}.  In practice, demonstrating the presence of topological order based on numerical results for finite-size systems can be a challenge. In Ref. \cite{Zhang2011}, the authors used Monte Carlo techniques to calculate the topological entanglement entropy and concluded in favor of topological order in  a projected BCS wave function, but later work \cite{ZhangErratum,pei2013topological} found that the result deviates  significantly from the expected value for a $\mathbb Z_2$ QSL. 

In this work, we construct  projected fermionic states for which we can rigorously establish or rule out topological order. The idea is to employ   stabilizer codes \cite{zeng2019quantum}, which can be defined for  spins as well as  fermions. For spins, Kitaev's toric code \cite{kitaev2003fault} and Wen's plaquette model \cite{wen2003quantum} represent topologically ordered stabilizer states. For fermionic systems, the formalism was developed in Ref. \cite{bravyi2010majorana} in the form of Majorana  fermion codes. Here we consider parent Hamiltonians for free fermions on lattices with four Majorana fermions per unit cell, which can be written as a Bogoliubov-de-Gennes Hamiltonian for spin-1/2 fermionic partons. We determine the precise conditions that guarantee that the projection of a Majorana stabilizer state leads to a topologically ordered state whose anyons satisfy the toric code fusion rules. Going beyond these exact results, we use MPS  to study projected fermionic states that interpolate between trivial and topologically ordered projected Majorana fermion codes.  We compute  the topological entanglement entropy and  find a topological transition in the projected state as we vary the state parameters, even though the  unprojected fermionic state  remains gapped along the interpolating path.  

The remainder of this paper is organized as follows. In Sec. \ref{sec:partons}, we give an overview of parton mean-field theories and the Gutzwiller projection. In Sec.  \ref{sec:mfcprojection}, we describe the formalism of stabilizer codes,  both for spin states and Majorana fermions. Section \ref{sec:projectiondatum} contains our  exact results concerning the topological order in projected Majorana fermion codes. In Sec. \ref{sec:numericalprojection}, we discuss the interpolation path between projected Majorana fermion codes and present our numerical results for the topological phase transition. Our concluding remarks can be found in Sec. \ref{sec:conclusion}. Detailed proofs of the lemmas and the  theorem in Sec. \ref{sec:projectiondatum} are relegated to the appendices.

\section{$\mathbb Z_2$ Partons \label{sec:partons}}
 
In this section, we briefly  review the basic aspects  of parton mean-field theories, especially in the context of $\mathbb Z_2$ QSLs, elucidating the problem that  motivates our work.

Consider a  two-dimensional  lattice $\Lambda$, with sites labeled by $i\in \Lambda$. On each site we place a spinful fermionic variable $f_{i \alpha}$, with spin index $\sigma\in \{\uparrow,\downarrow\}$. A general  free-fermion Hamiltonian in this setup can be written as
\be
H_\mathrm{ff} = \sum_{\langle i, j \rangle}\sum_{\alpha, \beta }(t^{\alpha \beta}_{ij}f^\dagger_{i\alpha}f^{\phantom\dagger}_{j \beta} + \Delta^{\alpha \beta}_{ij} f^\dagger_{i\alpha} f^\dagger_{j\beta} + \mathrm{H.c.})\;,
\label{eq:freefermionZ2}
\ee
where $\braket{i,j}$ denotes nearest-neighbor bonds. In addition to spin-dependent hoppings $t^{\alpha \beta}_{ij}$,  we   allow terms that break  fermion number conservation by incorporating pairing amplitudes $\Delta^{\alpha \beta}_{ij}$. 

The main goal of this approach  is to construct  a spin state out of the ground state of $H_\mathrm{ff}$, denoted as $|\psi_\mathrm{ff}\rangle$. We introduce spin-1/2  Abrikosov fermions as   \cite{savary2016quantum}
\begin{equation}
    \boldsymbol{\sigma}_i = f^\dagger_{i \alpha}\boldsymbol{\tau}_{\alpha\beta}f^{\phantom\dagger}_{i\beta} \;,
    \label{eq:abrikosovpartons}
\end{equation}
where $\boldsymbol \tau = (\tau^x, \tau^y, \tau^z)$ represents the respective Pauli matrices with matrix indices $\alpha,\,\beta$. Clearly,   there is a mismatch between the two representations since  the fermionic Hilbert space at each site is four-dimensional, consisting of states with  a single  (spin-up or spin-down) fermion or   states with no fermions or a double occupancy. To select the spin-1/2 states, we  project out   the other two states using the Gutzwiller projection
\begin{equation}
    |\psi_G \rangle  = P_G|\psi_\mathrm{ff}\rangle \equiv \prod_i n_i (2-n_i) |\psi_\mathrm{ff}\rangle\;,
    \label{eq:gutzwillerproj}
\end{equation}
where $n_i = \sum_{\alpha \in \{\uparrow, \downarrow\}} f^\dagger_{i\alpha}f^{\phantom\dagger}_{i\alpha}$ is the fermion number at site $i$. Within the image of  $P_G$, the    operators defined in Eq. (\ref{eq:abrikosovpartons})  obey the usual SU(2) algebra, and we  denote the physical Pauli operators by $\boldsymbol{\sigma}_i=(X_i,Y_i,Z_i)$. 

The  low-energy properties of $|\psi_G \rangle$ can be analyzed by embedding  the problem  into a $\mathbb Z_2$ gauge theory; see Ref. \cite{sachdev2023quantum} for   details. The resulting theory  is inherently strongly interacting. The mean-field approximation becomes controlled if one generalizes the model to contain $N$ fermion flavors and takes the limit $N\to \infty$.  The more physically relevant case, $N= 2$, can be explored by numerical methods \cite{becca2017quantum,Jin2020,wu2020tensor, petrica2021finite}. 

The Gutzwiller projection establishes a relation between spin phases and free-fermion phases, meaning that one can use the known categories of fermionic states to classify the different QSL  phases \cite{sachdev2023quantum}. In particular, a BCS supercondutor classified as  a trivial gapped phase  \cite{Chiu2016} corresponds, in the parton theory, to a $\mathbb Z_2$ QSL.  This represents a profound transformation in the nature of the state: The Gutzwiller projection converts a trivial, short-range-entangled state into a long-range-entangled state characterized by anyonic excitations. From the perspective of the gauge theory, this phenomenon is attributed to the deconfinement of the associated charges. However, the projection can also lead to a trivial spin state, which in the gauge theory corresponds to strong gauge fluctuations that invalidate the   mean-field approximation and lead to the confinement of the fermionic  partons into conventional bosonic spin excitations.

\section{Stabilizer codes \label{sec:mfcprojection}}
In this section, we   develop the tools needed to construct    models in which  the properties of the projected states can be exactly computed. Readers familiar with topological stabilizer codes can skip the first subsection. 

\subsection{Topological stabilizer codes\label{topcodes}}

In this section, we review the properties of stabilizer states and codes \cite{gottesman1997stabilizer}, which play a central role in our construction, with a particular emphasis on results related to topological codes.

The main object to be defined is the Pauli group of $n$ qubits, denoted by $\mathcal P_n  = \langle \{iI, X_j, Y_j, Z_j\}_{j=1}^n \rangle$,  consisting of all $n$ qubit Pauli strings. Here  $I$ is the identity and $\langle \mathcal G \rangle $ denotes   the group generated by a set $\mathcal G$.  A stabilizer group is an Abelian subgroup $\mathcal S \subseteq \mathcal P_n$ with $-I \notin \mathcal S$. These groups are particularly significant because they define a \emph{stabilizer code}:
\be
V_\mathcal S  \equiv \{|\psi\rangle \in (\mathbb C^2)^{\otimes n} \;|\; s |\psi\rangle = |\psi\rangle,\;\forall s \in \mathcal S\}\;,
\label{eq: stabilizercode}
\ee
which represents the set of all states stabilized by the operators in $\mathcal S$. The name comes from the intuiton that we are encoding $n$ qubits into $k $ logical qubits  \cite{nielsen2002quantum}. Here $k$ is called the encoding rate and is  given by\be
k= \log_ 2 \mathrm{dim}(V_\mathcal S).
\ee 

Another important measure of a code is its robustness. Let $\ell$ be a logical operator that causes  transitions between  two orthogonal code states $\ket{i},\,\ket{j} \in V_\mathcal S$,  i.e., $\ell \ket{i} \propto \ket{j}$. The action of these logical operators can be interpreted as creating errors (or excitations) in the code.  The code distance $d$ corresponds to the minimal ``size'' of the operator that is needed to mix the code states. More precisely,   consider the set $\mathcal L$ of logical operators of $\mathcal S$, constructed as the centralizer of $\mathcal S$ within $\mathcal P_n$ modulo the stabilizer elements, i.e.,  $\mathcal L = \mathrm C_{\mathcal P_n}(\mathcal S) \backslash \mathcal S$. The code distance is defined as
\be
d =  \min_{\ell  \in \mathcal L }|\mathrm{supp}(\ell)|\;,
\label{eq:distanceform}
\ee
where $|\mathrm{supp}(\ell)|$ is the number of sites in the support of $\ell$. Codes with larger values of  $d$ have a natural notion of robustness, as the states within the code can only be  mixed   at high orders in a perturbative analysis of local interactions. 

Various kinds of gapped phases can be seen   as stabilizer codes.  For instance, the $\mathbb Z_2$-symmetry-breaking  phase of the Ising  chain can be described by the code stabilized by $\mathcal S_\mathrm{Ising} = \langle \{Z_i Z_{i+1}\}_{i=1}^n\rangle$, which stabilizes the subspace spanned by $|0\rangle^{\otimes n}$ and $|1\rangle^{\otimes n}$, with $k=1$. %\racom{There is also the example of the cluster state in various dimensions, which forms the representative of $\mathbb Z_2 \times \mathbb Z_2$ SPTs. In the interest of space, I prefer not to mention it.}
 Another notable example is Kitaev's toric code \cite{kitaev2006topological}, defined by the stabilizer group $\mathcal S_\mathrm{TC} = \langle \{X_i X_j X_k X_l\}_{ijkl \in \square}, \{Z_i Z_j Z_k Z_l\}_{ijkl \in +}\rangle$. Here, qubits reside on the edges of a square lattice embedded on a torus, with $\square$, $+$ referring to the plaquettes and vertices of the lattice, respectively. The toric code is an exactly solvable example of a $\mathbb Z_2$ topologically ordered state, since its encoding rate only depends on the topology of the manifold where it is defined. 

Another  useful feature of stabilizer codes is that they can    be naturally described using a Hamiltonian framework. For every stabilizer group $\mathcal{S}$, we can construct a parent Hamiltonian
\be
H = -\sum_{g \in \mathrm{gen(\mathcal S)}}g\;,
\label{eq:Hstab}
\ee
where $\mathrm{gen}(\mathcal S)$ is the generating set of $\mathcal S$. For the two examples mentioned above, the corresponding Hamiltonians are, as expected, the classical Ising chain and the   toric code Hamiltonian. The encoding rate $k$ is related to the ground state degeneracy by  $\text{GSD} = 2^k$. It can be obtained from the relation $k = n-|\mathrm{gen}(\mathcal S)|$, where $|\mathrm{gen}(\mathcal S)|$ is the number of generators. Furthermore, if a perturbation $V$ is added to the Hamiltonian, so that the ground states are no longer stabilizer states, it can be rigorously shown that the order of perturbation theory at which  tunneling between the code states occurs scales   with the distance \cite{bravyi2010topological}, supporting the interpretation of $d$ as a robustness measure. 

Usually, the code data are organized into a tuple, referring to $\mathcal S$ as a $[[n,k,d]]$ stabilizer code. For the Ising chain, any  $Z_j$ operator couples the GHZ states $\frac1{\sqrt2}(|0\rangle^{\otimes n}+|1\rangle^{\otimes n})$ and $\frac1{\sqrt2}(|0\rangle^{\otimes n}-|1\rangle^{\otimes n})$ in the code. Thus, $d=1$, and we have  a $[[n,1,1]]$ code. For the toric code on a two-dimensional surface with genus $g$,  any Pauli operator that  acts nontrivially within the ground state manifold   must  wrap its support around the possible cycles of the surface, which leads to a code distance   $d= \sqrt n$. The topological degeneracy is $4^{g}$,  corresponding to  the encoding rate   $k=2g$. From this perspective, the toric code is identified as a $[[n,2g,\sqrt n]]$ code.

One can argue that the toric code and its corresponding phase represent the   most robust stabilizer quantum liquid in two dimensions. This claim is supported by two key results. First, it has been shown \cite{brayvi2010tradeoffs} that, for a geometrically local stabilizer group whose generators (and thus the Hamiltonian) are defined on a  two-dimensional Euclidean lattice,  the corresponding $[[n,k,d]]$ code satisfies $k d^{2} \leq c n$, where $c=O(1)$ is a constant. For codes with $k=O(1)$, also referred to as gapped quantum  liquids \cite{zeng2015gapped}, it follows that the maximum code distance is $O(\sqrt n)$.  Second,  it was  demonstrated in Ref. \cite{haah2016classification} that any two-dimensional translationally invariant stabilizer group that saturates this bound is equivalent, up to finite-depth local unitaries, to copies of the toric code.
While  the  code distance may depend on the specific  geometry if we consider other two-dimensional lattices, as we will do in Sec. \ref{sec:projectiondatum}, the  encoding rate and anyonic excitations of this stabilizer quantum liquid remain unaffected.

\subsection{Majorana fermion codes\label{MFC}}

Fermionic systems are   ubiquitous in physics. Therefore, a natural question is whether  the stabilizer formalism  can also be applied to such systems. This generalization was pioneered in  Ref. \cite{bravyi2010majorana} using Majorana fermions.  Here we build on this theory adapting the notation for the implementation of the Gutzwiller projection to be discussed in Sec. \ref{sec:projectiondatum}. 

Consider Majorana operators $\gamma^a_{j}$, where $j\in\{1,\dots,n\}$ is the site index  and  $a\in\{0,\,\dots,\,2N_F-1\}$ labels  the fermionic species at each site. The  local Hilbert space  at each site is equivalent to  the Fock space $\mathcal F_{N_F}$ of $N_F$ complex fermions.  The Majorana  operators  are characterized by the Clifford algebra
\be
\left\{\gamma^a_{j}, \gamma^b_{j'}\right\} = 2\delta^{ab}\delta_{jj'}.
\ee
 An analog of the Pauli group can be defined by first introducing the group of Majorana operators, $\mathrm{Maj}(2N_Fn) = \langle \{iI, \gamma^0_j, \gamma^2_j, \cdots, \gamma^{2N_F-1}_j\}_{j=1}^n \rangle$. A subtle point regards the nonlocality of operators due to the anticommutation rules on different sites. To ensure locality, we     restrict ourselves to the parity-even subgroup, $\mathrm{Maj}^+(2N_Fn) \subseteq \mathrm{Maj}(2N_Fn)$, comprising  products of an even number of Majorana operators. Such products trivially commute on different sites, and are hence local in a standard fashion. 

Let us define the Majorana stabilizer group as an Abelian subgroup $\mathcal M \subseteq \mathrm{Maj}^+(2N_Fn)$  with $-I \notin \mathcal M$. The space
\be
V_\mathcal M \equiv \{|\psi\rangle \in  \mathcal F_{N_Fn} \;|\; m |\psi\rangle =|\psi\rangle, \; \forall m \in \mathcal M\}\;,
\ee
is called the Majorana stabilizer code, in analogy with the definition in Eq.  \eqref{eq: stabilizercode}. Its encoding  rate is given by $k=N_Fn-|\mathrm{gen}(\mathcal M)|$. Moreover,  the code distance is defined as in Eq. (\ref{eq:distanceform}), with the optimized set being $\mathcal L = \mathrm{C}_{\mathrm{Maj}^+(2N_Fn)}(\mathcal M) \backslash \mathcal M$, describing the set of operators which act nontrivially in the code. As before, it is convenient to express the code data as $[[n,k,d]]_F$, where the $F$ subscript indicates that the code is built out of the Majorana stabilizer group $\mathcal M$.

Given a group of Majorana operators, we can  construct a   parent Hamiltonian whose ground states are in a (free-fermionic) trivial gapped phase  and correspond to a code space $V_\mathcal{M}$. For this purpose, consider a Majorana stabilizer group $\mathcal M \subseteq \mathrm{Maj}^+(2N_Fn)$ that  satisfies the  following two conditions: ({\it i}) \emph{free-fermion}, i.e., all of its generators correspond to  Majorana fermion bilinears; ({\it ii}) \emph{nondegenerate}, i.e., the code has  encoding rate $k=0$. Enforcing these conditions allows us to define the corresponding Majorana stabilizer group
\be
\mathcal M = \langle \{i\gamma^I \gamma^J\}\rangle\;,~~~{(I,J) \in {[2N_Fn] \choose 2}}\,,
\ee
where the index $I=(j,a)$ combines position and flavor degrees of freedom and ${[2N_Fn] \choose 2}$ is a set of pairings of the $2N_Fn$ Majorana modes. 

%We could also write the corresponding stabilizer Hamiltonian using $N_F$ complex fermions defined from linear combinations of Majorana fermions. In the complex fermion representation, this Hamiltonian has  the standard Bogoliubov-de-Gennes form as written in Eq. (\ref{eq:freefermionZ2}).

To illustrate the construction, consider a  chain with $n$ sites, periodic boundary conditions,  and $N_F =1$, hence two Majorana modes per site,  which we denote as $\gamma_j^0$ and $\gamma_j^1$. We can then construct a stabilizer group $\mc M_\mathrm{Kitaev}$ as $\mathcal M_\mathrm{Kitaev} = \langle \{i\gamma^1_j \gamma^0_{j+1}\}_{j=1}^n\rangle \subseteq \mathrm{Maj}^+(2n)$. It is straightforward to check that this choice obeys all the conditions of a Majorana stabilizer group. The corresponding parent Hamiltonian is simply the standard Majorana chain, also known as the Kitaev chain \cite{kitaev2001unpaired}.

We are mainly interested in the case $N_F=2$ on two-dimensional lattices, which can be connected with the parton mean-field theory for a $\mathbb Z_2$ QSL. In this case, we   define two complex fermions at each site by taking linear combinations of the Majorana modes:\begin{eqnarray}
f_{j\uparrow}=\frac12(\gamma_j^0+i\gamma_j^3),\nonumber\\
f_{j\downarrow}=\frac12(\gamma_j^2-i\gamma_j^1).\label{ffermions}
\end{eqnarray}
The single-occupancy constraint in terms of $f$ fermions is equivalent to a parity constraint for Majorana fermions \cite{kitaev2006anyons,burnell2011su2}:\begin{equation}
D_j=\gamma^0_j \gamma^1_j \gamma^2_j \gamma^3_j = 1\qquad \forall j. \label{parityconstr}
\end{equation}
Within the subspace of states that respect this local constraint,  the physical Pauli operators are represented as\begin{equation}
\boldsymbol\sigma_j=(i \gamma^1_j \gamma^0_j,\,i \gamma^2_j \gamma^0_j,\,i \gamma^0_j \gamma^3_j). \label{eq:kitaevpartons}
\end{equation}
Given a state in the Majorana stabilizer code, we   obtain a spin wave function by applying a   Gutzwiller projector analogous to  Eq. \eqref{eq:gutzwillerproj}:\be
P_G |\psi\rangle = \prod_j \left(\frac{1+D_j}{2}\right)|\psi\rangle\;.
\ee

As already shown in Ref.  \cite{bravyi2010majorana}, this  relation  between Majorana fermions and Pauli operators  implies   that every $[[n,k,d]]$ qubit stabilizer code can be ``fermionized" into a $[[4n,k, 2d]]_F$ Majorana stabilizer code. Given  a qubit stabilizer group $\mathcal S$, we construct a Majorana stabilizer group by: ({\it i}) substituting every Pauli by the parton formula  in Eq. (\ref{eq:kitaevpartons}); ({\it ii}) adding $D_j = \gamma^0_j \gamma^1_j \gamma^2_j \gamma^3_j$ as a generator for every site, building the group $\mathcal M_\mathcal S$. It can be verified that all $4n$ modes are locked into a code with the same rate, while  the distance is doubled because each Pauli operator is a Majorana bilinear.

\section{Topological properties of projected Majorana fermion codes\label{sec:projectiondatum}}

While it is remarkable that a qubit stabilizer group can be fermionized, our focus is on the reverse process. Specifically, given a free-fermion, nondegenerate Majorana stabilizer group, we investigate  the properties of  the corresponding qubit stabilizer code. Similar ideas were discussed in Refs. \cite{Hassler2012,Vijay2015,Litinski2018,you2019higher,lensky2023graph,hastings2023quantum}, where it was proposed that the local constraint  in Eq. (\ref{parityconstr}) can be dynamically implemented  as an interaction term in the Hamiltonian. Physically, this interaction can arise from the charging energy in Majorana-Cooper boxes, which are experimentally feasible platforms for Majorana surface codes \cite{Landau2016}. Rather than discussing a specific Hamiltonian, here we  focus on the Gutzwiller projected state and aim to establish the precise  conditions for it  to be topologically nontrivial.

We start by   imposing that $\mathcal M$   be {parity-even}, meaning that $\prod_j D_j \in \mathcal M$. Otherwise, the Gutzwiller projection would trivially vanish if $\prod_j D_j =-1$. Based on the  discussion in Sec. \ref{MFC}, we  concentrate  on subgroups of $\mathrm{Maj}(4n) = \langle i, \{\gamma^0_j, \gamma^1_j, \gamma^2_j, \gamma^3_j\}_{j=1}^n \rangle $. We then have  the following lemma:

\begin{lemma}
 Let $\mathcal M \subseteq \mathrm{Maj}^+(4n)$ be a nondegenerate, free-fermion, parity-even Majorana stabilizer group with code \{$|\psi_\mathcal M\rangle$\}. Then, there exists a qubit stabilizer code $P_G\mathcal M \subseteq \mathcal P_n$, stabilizing the state $|P\psi_\mathcal M\rangle$, such that, given an operator  $m \in \mathrm{Maj}(4n)$ that commutes with all local parities $D_j$, the code state obeys 
 \begin{equation}
     \frac{\bra{\psi_\mathcal M}P_G m P_G \ket{\psi_\mathcal M}}{\bra{\psi_\mathcal M}P_G \ket{\psi_\mathcal M}} = \bra{P \psi_\mathcal M} \mathrm{P}(m)\ket{P \psi_\mathcal M}\;,
 \end{equation}
where $\mathrm P(m) \in \mathcal P_n$ is the corresponding Pauli operator according to the parton map   in Eq. (\ref{eq:kitaevpartons}).
 \label{Lemma: 1}
 \end{lemma}

\begin{figure}
    \centering
    \includegraphics[width=0.95\linewidth]{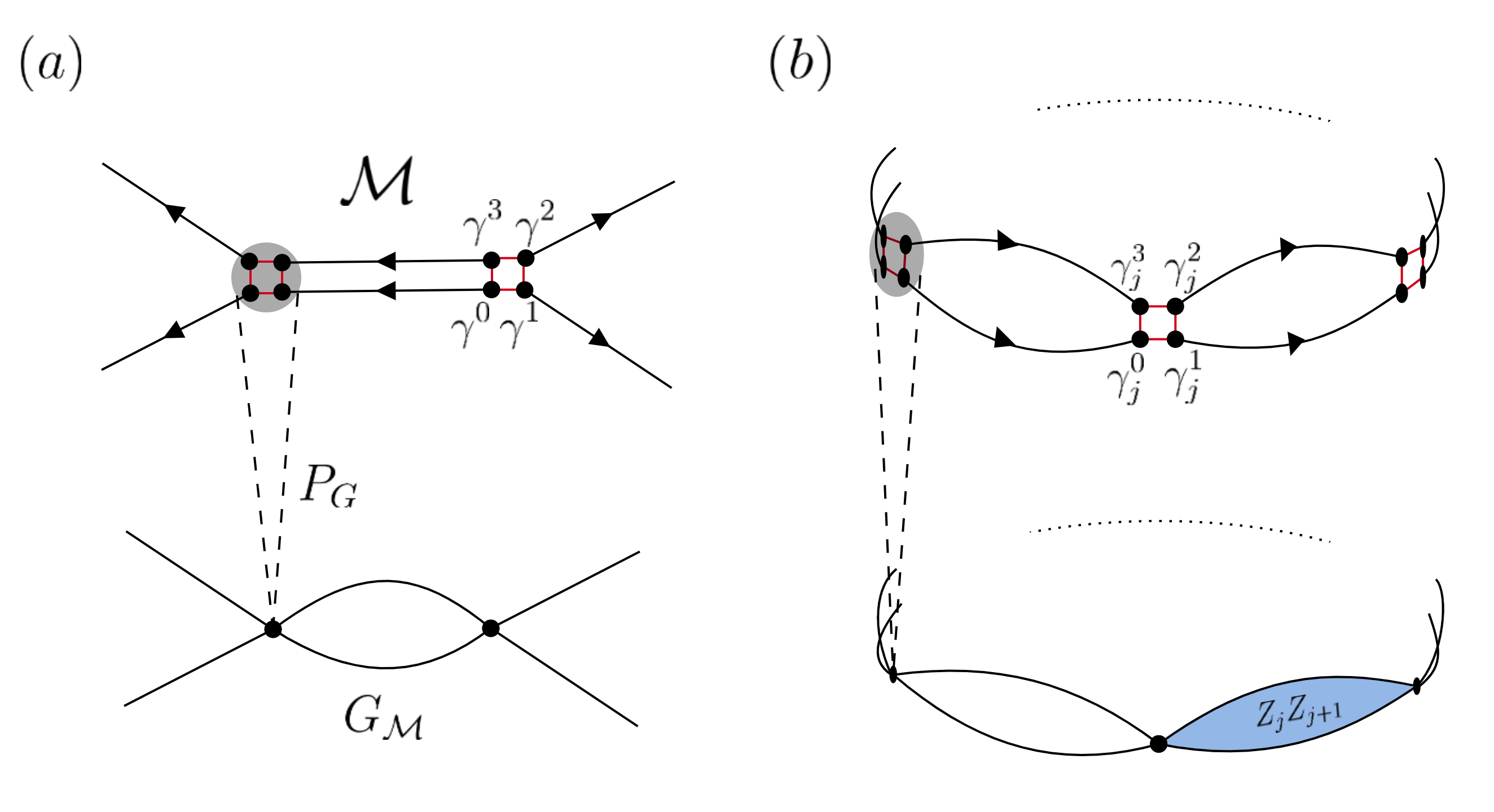}
    \caption{(a): Depiction of the Gutzwiller projection. The top image illustrates the oriented dimers representing the Majorana bilinears in $\mathcal M$, while the bottom image depicts the resulting projected graph $G_\mathcal M$. (b) Representation of the stacking of the two Majorana chains being projected into the Ising stabilizer code. The resulting projected state corresponds to the symmetric  GHZ state, locally stabilized by $Z_j Z_{j+1}$, as highlighted in the lower half of the figure. }
    \label{fig:projection}
\end{figure}

The proof is presented in Appendix \ref{app:proof}.   Lemma \ref{Lemma: 1}  establishes   the necessary conditions to properly bosonize a free fermion Majorana code into a qubit stabilizer code. The procedure is illustrated in Fig. \ref{fig:projection}(a). We consider  free-fermion Majorana stabilizer codes with four Majorana modes per site. Each generator $i\gamma^I\gamma^J\in \text{gen}(\mathcal M)$ is represented by a dimer with an arrow pointing from $I$ to $J$. After the projection, the corresponding qubit stabilizer code lives on the lattice formed by ``merging'' the Majorana modes of a site to form a qubit. Note that the operator $m$  in the expectation value $\langle \psi_\mathcal M| P_G m P_G |\psi_\mathcal M\rangle$ can also correspond to products of Majorana bilinears at different positions. Thus, the lemma includes the statement that correlation functions in the spin state, given by $\langle P\psi_\mathcal M|\mathrm P(m) |P\psi_\mathcal M\rangle$, are exactly reproduced by the calculation of the corresponding correlation on the fermionic side in the parton description.

The simplest example is obtained by considering a code in which  all Majorana  modes are paired within  the same site: $\mathcal M_\mathrm{trivial} = \langle \{i \gamma^0_j \gamma^3_j , \; i\gamma^2_j \gamma^1_j\}_{j=1}^n\rangle$. It is straightforward to  verify  that $D =\prod_j D_j \in \mathcal M_{\mathrm{trivial}}$ and $\mathcal M_\mathrm{trivial}$ satisfies  the conditions of Lemma \ref{Lemma: 1}. Since the generators of $\mathcal M_\mathrm{trivial}$ commute with all the local parities,  we can  apply Eq. (\ref{eq:kitaevpartons}) directly to obtain  $P_G \mathcal M_\mathrm{trivial} = \langle \{Z_1, \cdots, Z_n\}\rangle$. As a result, the stabilizer state is  the product state $|0\rangle^{\otimes n}$. It is easy to see that every Majorana code in which  all Majorana modes are paired within the same site  stabilizes a product state.

In a less trivial example,   consider  the   fermionic code defined from  the direct sum of two copies of the Kitaev chain,  $\mathcal M_\mathrm{Kitaev}^{\oplus 2} \equiv \langle \{i \gamma^1_{j}\gamma^0_{j+1}\}_j\rangle \oplus \langle \{i \gamma^2_j \gamma^3_{j+1}\}_j\rangle = \langle \{i \gamma^1_{j}\gamma^0_{j+1}, i \gamma^2_{j}\gamma^3_{j+1}\}_{j=1}^n\rangle$.  This example is illustrated in Fig. \ref{fig:projection}(b).  Since the   generators  contain only one Majorana per site,  they do not commute with the local parities, but the product of two generators on the same bond does. Applying the parton mapping in Eq. (\ref{eq:kitaevpartons}), we obtain  $\text{P}(i \gamma^1_j \gamma^0_{j+1}\,i \gamma^2_j \gamma^3_{j+1}) = Z_{j}Z_{j+1}$. Moreover,  the fermion  parities  of each Majorana chain,   $\prod_j i \gamma^1_j \gamma^0_{j+1}$ and $ \prod_j i \gamma^2_j \gamma^3_{j+1}$,  are  mapped onto the global $\mathbb{Z}_2$ symmetry operator  $\prod_j X_j$ of the Ising  chain. The corresponding projected group is $\mathcal S_\mathrm{GHZ} \equiv P_G \mathcal M_\mathrm{Kitaev}^{\oplus 2} = \langle \{Z_j Z_{j+1}\}_j \cup  \{\prod_j X_j\}\rangle$ which stabilizes the $\mathbb{Z}_2$-symmetric GHZ  state $|\psi\rangle = (|0\rangle^{\otimes n} + |1\rangle^{\otimes n})/\sqrt 2$. Here we stress that something quite nontrivial happened in this procedure. The ground state of the Kitaev chain is an example of a short-range entangled state, which means that it can be connected to a (trivial)  product  state under the action of a finite-depth quantum circuit  \cite{Chen2010}. However, the  GHZ state  exhibits long-range entanglement.  This  simple  example illustrates the power of the projection in  generating nontrivial states from trivial ones. %This can be contrasted with the notion of sequential quantum circuits \cite{Chen2024}, where a linear depht circuit serve as a map between different area-law states. In our construction, this role is played by the projection, which modifies the entanglement structure of the original state. 

Up to this point, there is no notion of geometrical locality on the codes, meaning that the Majorana modes can be paired arbitrarily without any restraint on the possible geometries in which they live. In fact, we can exploit this arbitrariness to explore higher-dimensional codes. To achieve this, the key condition we have to impose on the Majorana dimers is that they tessellate the surface of interest. We now discuss how one can achieve that explicitly.

Given a free-fermion, nondegenerate,  parity-even Majorana stabilizer group $\mathcal M$, we want to construct the multigraph $G_\mathcal M = (V,E)$   where the projected qubit stabilizer code is defined, as  illustrated in the lower panel of  Fig \ref{fig:projection}(a). We identify the  vertices  of $G_\mathcal M$    with the sites, $V=\{1,2,\cdots, n\}$, and  associate edges $E$ with pairs of vertices   connected by Majorana dimers.  However, this condition  alone is not sufficient to uniquely specify  a given edge since there can be  more than one Majorana dimer connecting a pair of vertices. To remedy this, let us   consider the set of supports of $\mathcal M$, given by the set
\be
\mathrm{supp}(\mathcal M) = \{(i, j) \in V^2 \,|\, \exists a,b : i \gamma^a_i \gamma^b_j \in \mathrm{gen}(\mathcal M)\}\,,
\ee
with $i\neq j$ and $V^2=V\times V$. An edge $E$ can then be defined by the following procedure: Given $i\gamma^a_i \gamma^b_j \in \mathrm{gen}(\mathcal M)$, we define the element  $((i,j), (a,b)) \in E$. We denote the resulting projected graph by  $PG_\mathcal M$. Note that the condition   $i \neq j$ excludes loops within the same site in $PG_\mathcal M$, but the projected graph can have two  edges between  a pair of sites  if two elements of $E$ share the same support on $V$.

We say that $\mathcal M$ is \emph{two-dimensional} if it satisfies all the conditions of Lemma \ref{Lemma: 1} and there is an oriented Riemann surface $\Sigma$ with metric $g$ and an graph embedding $\sigma: G_\mathcal M \to \Sigma$ \footnote{This graph embedding must be \emph{cellular}, that is, the faces must be homeomorphic to open disks. This allows us to properly compute topological notions.} preserving the distances. We will omit both the embedding and the metric and refer to the necessary data as $(\mathcal M, \Sigma)$.

As a concrete example,  consider  a square lattice formed by Majorana dimers with periodic boundary conditions, as shown in Fig. \ref{fig:square}. The corresponding Majorana stabilizer group is  $\mathcal M_\mathrm{Sq} = \langle \{i\gamma^1_{\mathbf r}\gamma^3_{\mathbf r+ \hat{\mathbf x}},i\gamma^2_{\mathbf r}\gamma^0_{\mathbf r+ \hat{\mathbf y}} \}_\mathbf r\rangle$, where $\mathbf r$ are positions on the square lattice with primitive vectors $\hat{\mathbf x}$ and $\hat{\mathbf y}$.  The manifold is a 2-torus, given by $\Sigma = T^2$, endowed with the standard Euclidean metric. The number of sites is $n = L_x L_y$, where $L_x$ and $L_y$ measure the lengths (in units of the lattice spacing) along $x$ and $y$ directions, respectively. We assume that $L_x$ and $L_y$ are both even. In this case, all Majorana modes are paired into dimers and $\mathcal M_\mathrm{Sq} $  satisfies the parity-even constraint.

\begin{figure}
    \centering
    \includegraphics[width=0.9\linewidth]{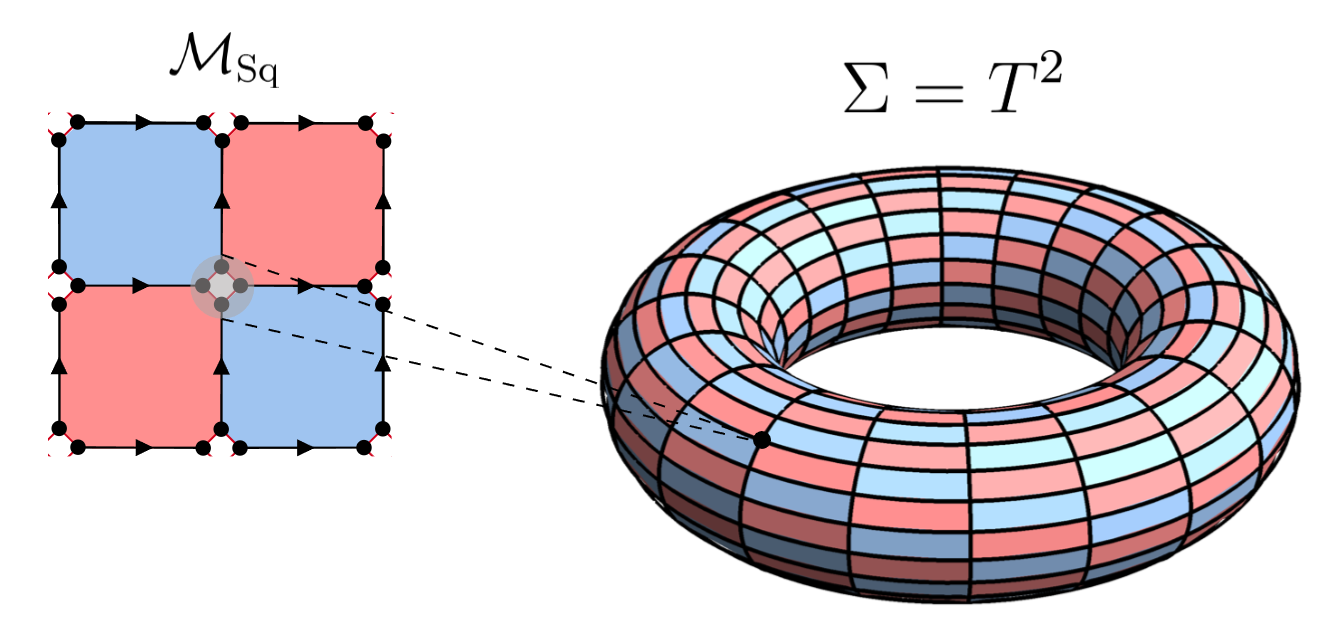}
    \caption{Illustration of the projection of $\mathcal M_\mathrm{Sq}$. We impose  an embedding of the graph  $G_\mathcal M$ into a   surface $\Sigma$,  reflecting the   geometrically local bonds. }
    \label{fig:square}
\end{figure}

To obtain the generators of the projected group $P_G \mathcal M_{\mathrm{Sq}}$, we have to take products of Majorana dimers that contain two modes per site. The minimal elements of $\mathcal M_\mathrm{Sq}$ that commute with the local parities are given by combinations of dimers forming a plaquette, which have the  form $(i\gamma^1_{i} \gamma^3_ {j})(i\gamma^2_ {j} \gamma^0_ {k})(i \gamma^1_{l} \gamma^3_{k})(i\gamma^2_ i\gamma^0_ l)$, where $( i  j  k  l)$  labels the four vertices in a plaquette corresponding to the positions $(\mb r,\mb r+\hat {\mb x},\mb r+\hat{\mb x}+\hat{\mb y},\mb r+\hat{\mb y})$; see  Fig. \ref{fig:square}. Under the parton map, this combination leads to the qubit operator $Z_ i X_ jZ_ k X_ l$. There is another set of nontrivial generators constructed by taking products of  dimers around one of the essential cycles of the 2-torus, given by $\prod_{\mathbf r \in h_1} i \gamma^2_\mathbf r \gamma^0_{\mathbf r+ \hat{\mathbf y}} $ and $\prod_{\mathbf r \in h_2} i \gamma^1_\mathbf r \gamma^3_{\mathbf r+ \hat{\mathbf x}} $, where $h_1$ and $h_2$ are the vertical and horizontal essential cycles, respectively. After the projection, they become the string operators  $  \prod_{\mathbf r \in h_{1}(h_2)}Y_\mathbf r$.

We then have that the projected group is $P_G\mathcal M_\mathrm{Sq} = \langle \{Z_  i X_  jZ_  k X_  l,\prod_{\mathbf r \in h_{1}}Y_\mathbf r,\prod_{\mathbf r \in h_{2}}Y_\mathbf r\}_{(  i   j   k   l), h_1, h_2}\rangle$. Keeping  only the local elements (plaquettes) in the generators group,   we obtain the parent Hamiltonian
\be
H_\mathrm{Wen} \equiv  - \sum_{(  i   j   k   l)} Z_  i X_  j Z_  k X_  l\;,
\label{eq:wenhamiltonian}
\ee
known as the Wen plaquette model \cite{wen2003quantum}. The latter is equivalent to the toric code Hamiltonian  under a Hadamard transformation on one of the sublattices \cite{Nussinov2009}.  Its ground state exhibits $\mathbb Z_2$ topological order and is fourfold degenerate on the torus. However, the additional string operators in  $P_G\mathcal M_\mathrm{Sq}$ select  one of the ground states, which obeys $\prod_{\mathbf r \in h} Y_\mathbf r |\psi\rangle = +|\psi\rangle$
and corresponds to the projection of the fermionic state  $|\psi_{\mathcal M_\mathrm{Sq}}\rangle$.

Note that the generators of the projected group come from \emph{cycles} of the projected graph $PG_\mathcal M$, since the projection enforces an even number of Majoranas at each site. We can distinguish between contractible cycles, defined by local products of Majorana operators acting on an $O(1)$ support, and essential cycles, which wind around the torus and involve an $O(\sqrt n)$ support.  This observation can be formalized as a Lemma:

\begin{lemma}
 Let $(\mathcal M, \Sigma)$ be a two-dimensional Majorana group. Then, it can be decomposed as the union $P_G\mathcal M = P_\partial \mathcal M \cup P_h \mathcal M$, where $P_\partial \mathcal M$ is a Majorana stabilizer group with support on boundaries of faces of $\sigma$, and $P_h\mathcal M$ is supported on essential cycles.
\label{Lemma: 2}
\end{lemma}

The proof of the above Lemma can be found in appendix \ref{prooflemma2}. In the square lattice example, $P_\partial \mathcal M$  corresponds to the stabilizer group of the Wen plaquette model, $P_\partial \mathcal M_\mathrm{Sq}  = \langle \{Z_  i X_  j Z_  k X_  l\}_{(  i   j   k   l)}\rangle$. Now, we claim that the embedding is \emph{two-colorable}, meaning that each plaquette (face) in the lattice can be labeled by a color index, say red (R) or blue (B), such that two plaquettes with the same color never share a common edge. This means that the generators split into two groups, according to their color. In the following, we refer to a two-colorable, two-dimensional Majorana group with a two-colorable embedding:

\begin{lemma}
    
Let $(\mathcal M, \Sigma)$ be a two-colorable, two-dimensional Majorana group. Then, the boundary Majorana group factorizes into the product $P_\partial \mathcal M \cong  P_R \mathcal M \times P_B \mathcal M$, where the generators of $P_{R/B}\mathcal M$ are supported on red and blue faces, respectively.
\label{lemma3}

\end{lemma}

\begin{proof}
     Since every generator is supported on a boundary of a face, it follows that every element $m \in P_\partial \mathcal M$ can be written as $m=m_Rm_B$, where $m_{R/B} \in P_{R/B}\mathcal M$. Since $\mathcal M$ is Abelian, this means that it is isomorphic to $P_R\mathcal M \times P_B\mathcal M$, whose elements are the pairs $(m_R,m_B)$.
\end{proof}

The structure of qubit stabilizer codes tesselated on surfaces is severely restricted, as   previously noted in Ref.  \cite{anderson2013homological}, where they were referred to as ``homological stabilizer codes''. As Lemma \ref{lemma3} states, for two-colorable embeddings, the local stabilizer group factorizes even further into two single-colored groups. We dub this feature as the code having a \emph{CSS structure}, since they are reminiscent of  Calderbank-Sloane-Shor codes, stabilizer codes whose generators are  built with either $X$s or $Z$s \cite{calderbank1996good, steane1996multiple}, splitting as $\mathcal S \cong \mathcal S_X \times \mathcal S_Z$. CSS codes are particularly useful to analyze and bound parameters, since the building blocks are equivalent to \emph{classical} error-correcting codes, for which  a mature theory already exists \cite{macwilliams1977theory}.

At this point, we are able to enunciate our main result:

\begin{theorem}
 Let $(\mathcal M, \Sigma)$ be a two-colorable, two-dimensional Majorana group. Then, $P_\partial \mathcal M$ stabilizes a $[[n,k,d]]$ stabilizer code with:
\begin{itemize}
    \item Topological encoding: $k$ is a topological invariant;
    \item String-like symmetries: There are logical operators supported on closed loops;
    \item Anyon excitations: The action of local operators excite one of $\{1,e,m,f\}$, satisfying the toric code fusion rules \cite{kitaev2003fault}.
\end{itemize}
\label{theorem1}
\end{theorem}
Thus,  the properties of the projected state match the expectation for $\mathbb Z_2$ topological order. The detailed proof is presented in Appendix \ref{proofT1}.   

In Ref. \cite{wen2003quantum}, it was noted that the topological nature of the encoding rate $k$ of the Wen plaquette model depends on having  even lengths in both directions, i.e., $L_x$, $L_y \in 2\mathbb Z$. This can be rephrased in terms of the bicolorability condition, which is well-defined only for systems with even lengths. In fact, by analyzing the action of local operators in the code, we  can   verify that the color excitations have a one-to-one correspondence with the toric code anyons, meaning that $e$ anyons are hosted on, say, red faces, and $m$ anyons on blue ones. This reinforces the previous statement of the equivalence between the Wen plaquette model and the toric code model. 

\begin{figure}
    \centering
    \includegraphics[width=0.9\linewidth]{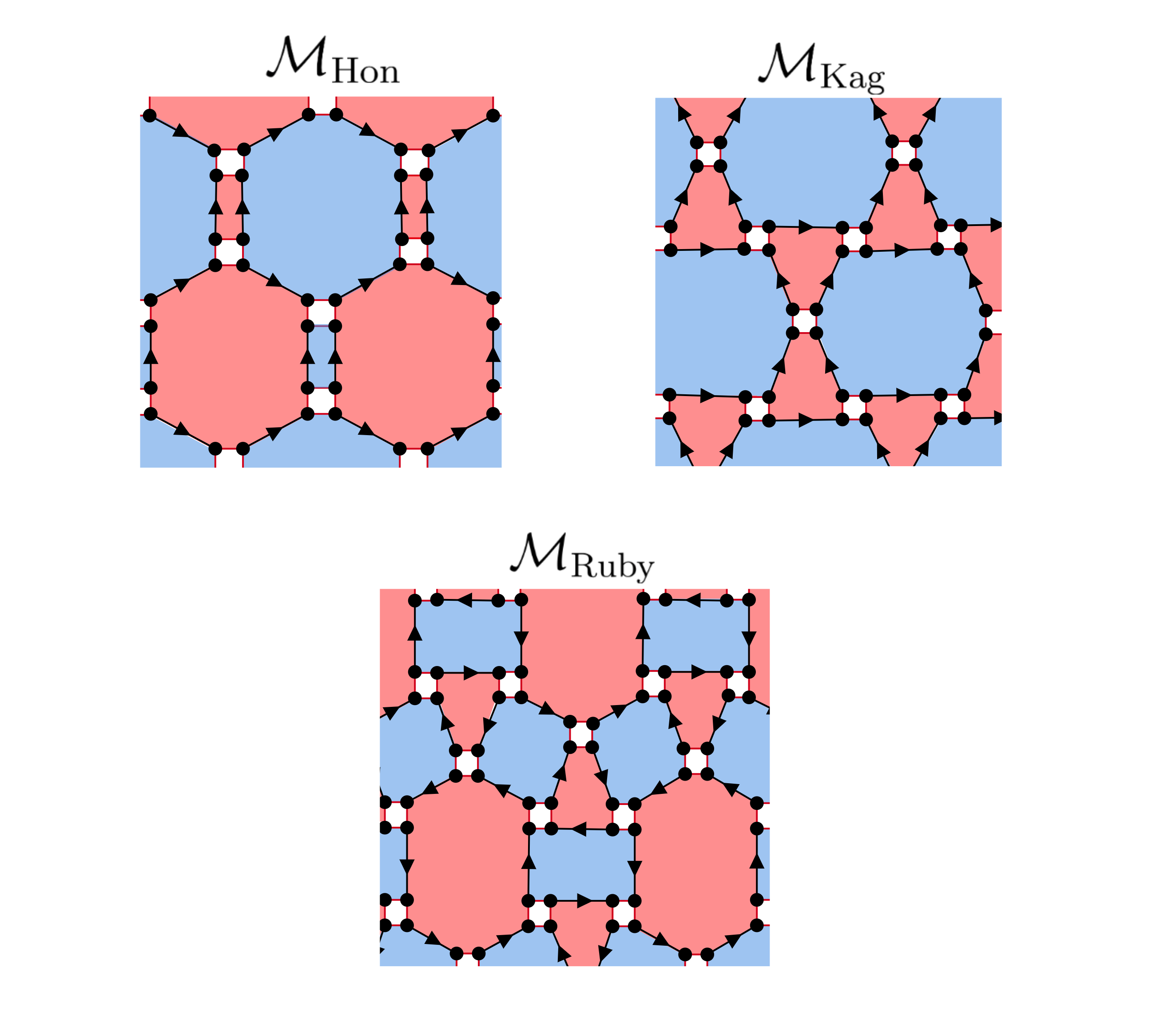}
    \caption{Majorana dimers on two-colorable graphs whose projections define  stabilizer codes on the honeycomb, kagome, and ruby lattices. In all these cases, the projected state has $\mathbb Z_2$ topological order.}
    \label{fig:otherexamples}
\end{figure}

While we have used the square lattice as our main example, the result of the theorem also applies to  other four-valent graphs. In Fig. \ref{fig:otherexamples}, we show examples corresponding to the honeycomb, kagome, and ruby lattices. The honeycomb example can  be derived from the highly anisotropic limit of Kitaev's honeycomb model \cite{kitaev2006anyons}. The kagome lattice example furnishes  a stabilizer code representative of the  $\mathbb Z_2$ topological order  discussed in Ref. \cite{verresen2021prediction}. We    note that the code distance depends on geometric aspects of the tesselation, as it scales with   the length of the nontrivial cycle of the surface $\Sigma$, as measured by the Riemannian metric $g$, usually referred to as the 1-systole \cite{quantum2021breuckman}. However,   all these examples can be understood as Euclidean tesselations and satisfy $d=O(\sqrt{n})$, in agreement with  the bound derived in Ref. \cite{delfosse2013tradeoffs}. The construction can even be applied to  hyperbolic codes with four-valent tesselations \cite{breuckmann2018phd} that satisfy the conditions of the theorem.

Another remark is that the toric/surface code topological order can host non-Abelian anyons by what is called twisting \cite{bombin2010topological}. This result was recently ``fermionized'' in Ref. \cite{lensky2023graph}, where the authors describe  the emergence of Ising anyons from Majorana models with $\mathbb Z_2$ gauge fields, in close relation to the setup studied here. This result  shows that, if some  Majorana  modes are left unpaired (i.e, a \emph{degenerate} Majorana fermion code in our language), the corresponding qubit code after the projection  can host Ising anyons.

\section{Topological phase transition in  projected fermionic states\label{sec:numericalprojection}}

In this section, we address the relation between free-fermion states and Gutzwiller projected states beyond the exactly solvable cases constructed from stabilizer groups.  First, we choose two free-fermion states $\ket{\psi_0}$ and $\ket{\psi_1}$, such that $  P_G \ket{\psi_0}$ is a $\mathbb{Z}_2$ QSL while $ {P}_G\ket{\psi_1}$ is topologically trivial. By showing that $\ket{\psi_0}$ and $\ket{\psi_1}$ are adiabatically connected, we provide an example in which  two fermionic states are in the same phase, but their Gutzwiller projections are not. Using MPS to compress the states along the adiabatic path, we are able to characterize the topological phase transition in the projected states by computing the entanglement entropy and spin correlation functions.

\subsection{Adiabatic paths between Majorana fermion codes}

Consider two fermionic states $\ket{\psi_0}$ and $\ket{\psi_1}$ which are ground states of free-fermion Hamiltonians $H_0$ and $H_1$, respectively. We say that $\ket{\psi_0}$ and $\ket{\psi_1}$ are in the same phase if there exists a parameterized Hamiltonian $H(s)$ such that  $H(0) = H_0$,   $H(1) = H_1$, and  $H(s)$ is gapped for all $s \in [0, 1]$ \cite{Hastings2005, Hastings2019}. In other words, there is an adiabatic path that connects the two states. Here we will  choose   $H_0$  and $H_1$ to be parent Hamiltonians for Majorana stabilizer codes on the square lattice.   

We start by defining  $H_0$ as
\begin{align}
H_0 & = - \sum_{\mb r} (i\gamma^1_{\mb r} \gamma^3_{\mb r+\hat {\mb x}} + i \gamma_{\mb r}^2 \gamma^0_{\mb r + \hat{\mb y}}).
\end{align}
We can rewrite this Hamiltonian using the complex spin-1/2 fermions defined in Eq. (\ref{ffermions}) as
\bea
H_0& =& -i\sum_{\mb r} (f^\dagger_{\mb r,\downarrow} f^{\phantom\dagger}_{\mb r + \hat {\mb y}, \uparrow} -f^\dagger_{\mb r,\downarrow} f^{\phantom\dagger}_{\mb r + \hat {\mb x}, \uparrow}  \nonumber\\
&&+f^{\phantom\dagger}_{\mb r,\downarrow} f^{\phantom\dagger}_{\mb r + \hat {\mb y}, \uparrow} +f^{\phantom\dagger}_{\mb r,\downarrow} f^{\phantom\dagger}_{\mb r + \hat {\mb x}, \uparrow}) + \mathrm{H.c.}, 
\eea
which has the form of the generic parton mean-field Hamiltonian in Eq. (\ref{eq:freefermionZ2}).  The spectrum of $H_0$ contains two flat bands at energies $\pm 2$. In terms of non-spatial symmetries \cite{Chiu2016}, the ground state $\ket{\psi_0}$ of $H_0$ is classified as a trivial gapped superconductor. As discussed in Sec. \ref{sec:projectiondatum}, the Gutzwiller projection of $\ket{\psi_0}$ yields a $\mathbb Z_2$ topologically ordered stabilizer state equivalent to the ground state of the Wen plaquette model. 

 \begin{figure}[t]
    \centering
    \includegraphics[width=1.0\linewidth]{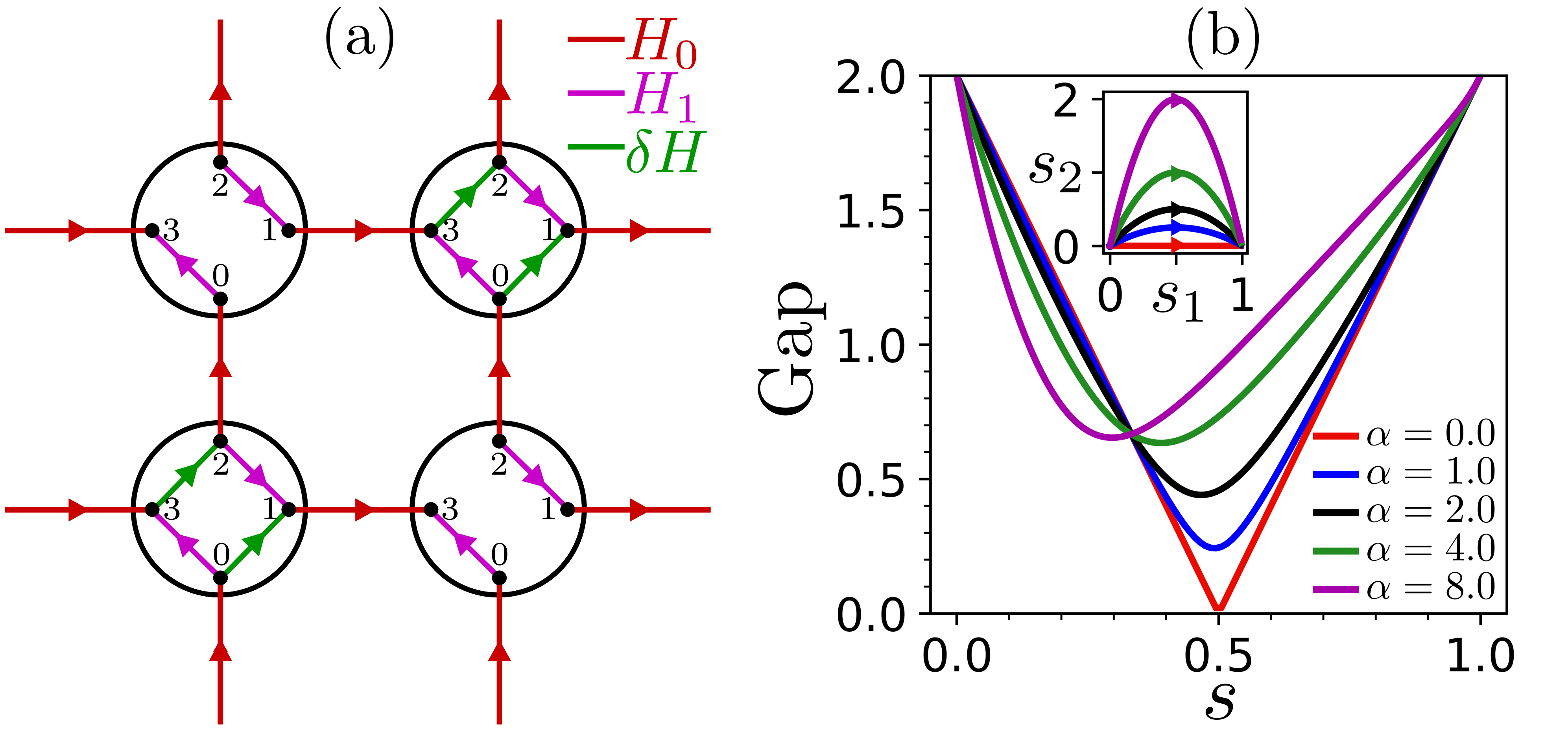}
    \caption{(a) Majorana fermion model on the square lattice. The different colors in the bonds represent the different bilinear  terms in the Hamiltonian in Eq. (\ref{Hpath}),  with red for $H_0$, purple for $H_1$, and green for $\delta H$.   (b) Energy gap along different paths on the $s_1 s_2$ plane, given by $s_1 = s$ and $s_2 = \alpha s(1-s)$ (inset).} 
    \label{fig:majorana}
\end{figure}

Next, we define $H_1$ as:
\begin{align}
H_1 & = - \sum_{\mb r} (i\gamma_{\mb r}^0 \gamma_{\mb r}^3 + i \gamma_{\mb r}^2 \gamma_{\mb r}^1).
\end{align}
In terms of complex fermions, we have \bea
H_1& =& -2 \sum_{\mb r} (f^\dagger_{\mb r,\uparrow} f^{\phantom\dagger}_{\mb r,\uparrow} - f^\dagger_{\mb r,\downarrow} f^{\phantom\dagger}_{\mb r,\downarrow}).
\eea
Clearly, $H_1$ is equivalent to a Zeeman term acting on the  fermionic partons. In this case, all Majorana fermions are paired within the same site, and the projected  state  $  P_G\ket{\psi_1}$ is a   product state in which  all spins are polarized along the $Z$ spin direction.  Hence, if we write a free-fermion Hamiltonian $H(s)$ that connects $H_0$ and $H_1$, such that $\ket{\psi(s)}$ is its ground state, then we expect a phase transition for $  P_G \ket{\psi(s)}$ at some $s \in [0,1]$, similar to the  transition in the toric code perturbed by a magnetic field \cite{Trebst2007,Dusuel2011}.  

We first try the simple linear interpolation $H(s) = (1 - s)H_0 + sH_1$. In this case, we find that the gap of $H(s)$ closes at $s=0.5$. This result can be understood by noting that the classification of free-fermion Hamiltonians is enriched by lattice rotation and translation symmetries. When these symmetries are taken into account,  the ground state of $H_0$ should be regarded a topological crystalline superconductor \cite{Teo2013}, while the ground state of $H_1$ remains trivial.  

To avoid the gap closing in the free-fermion model, we  generalize the Hamiltonian with a two-dimensional parametrization:
\begin{equation}
H(s_1, s_2) = (1 - s_1) H_0 + s_1 H_1 + s_2 \delta H,\label{Hpath}
\end{equation}
Here $\delta H$ acts  on only one of the two sublattices (denoted by $A$) and is given by
\bea
\delta H& = & - \sum_{\mb r \in A} (i\gamma_{\mb r}^0 \gamma_{\mb r}^1 + i \gamma_{\mb r}^3 \gamma_{\mb r}^2) \\
&= & 2 \sum_{\mb r \in A} (f^\dagger_{\mb r,\uparrow} f^{\phantom\dagger}_{\mb r, \downarrow} + f^\dagger_{\mb r,\downarrow} f^{\phantom\dagger}_{\mb r,\uparrow} ).\label{eq:zfield}
\eea
This new term  breaks translational symmetry and can be interpreted as a sublattice-dependent transverse  field.  Figure \ref{fig:majorana}(a) shows a schematic representation  of all the terms in $H(s_1, s_2)$ in the Majorana representation. 

\begin{figure}[t]
    \centering
    \includegraphics[width=1.0\linewidth]{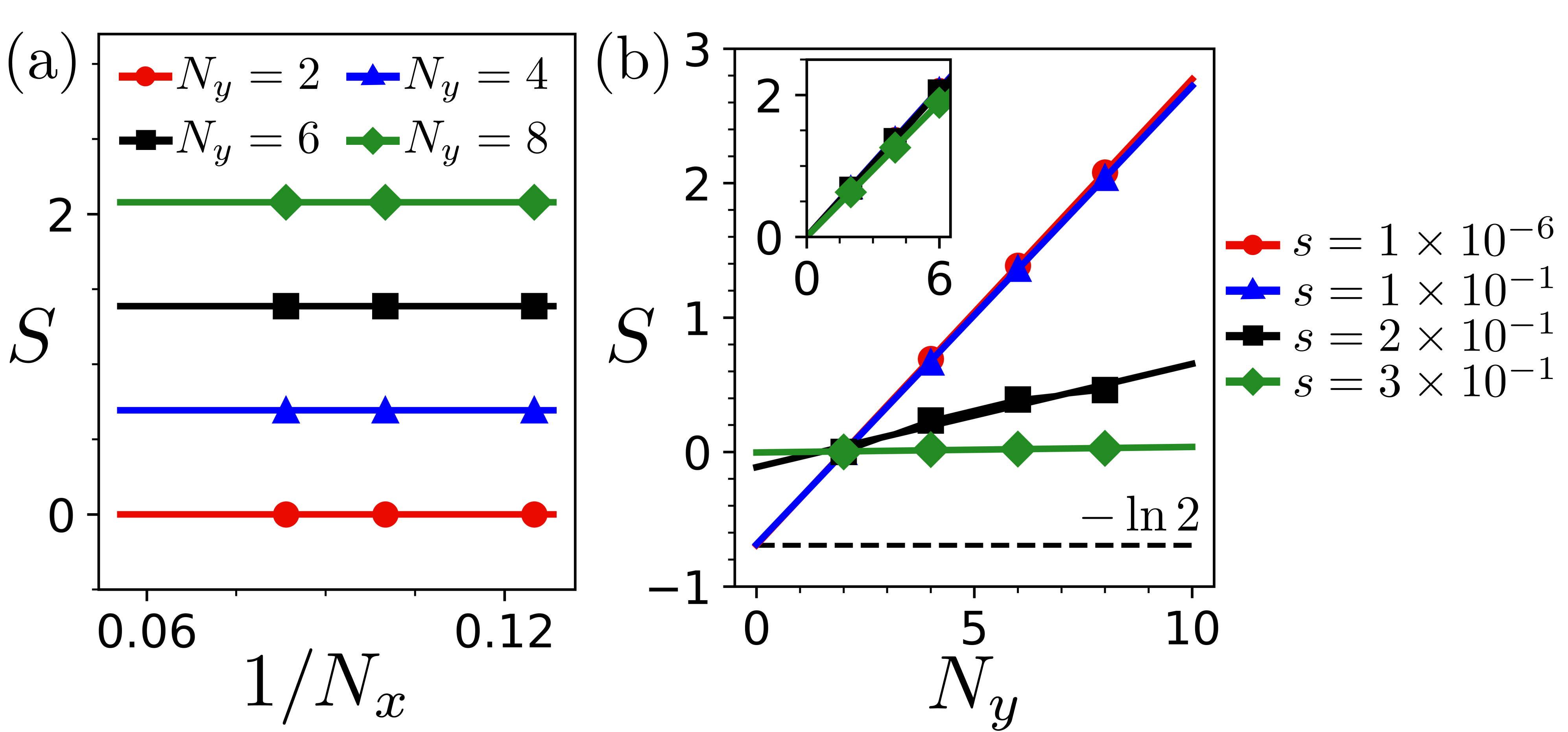}
    \caption{Entanglement entropy ($S$): (a) of $  P_G\ket{\psi(s)}$ as a function of  $1/N_x$ for several $N_y$ and $s = 1\times 10^{-6}$; (b) of $\ket{\psi(s)}$ (inset) and $ P_G\ket{\psi(s)}$ for $N_x \to \infty$ and $\alpha = 4$ as a function of $N_y$. Fitting the data with $S = aN_y - \gamma$, we get $\gamma = 0$ before the projection for all values of $s$. After the projection, we get  $\gamma = \ln 2$ (dashed line) for $s \leq 0.1$ and $\gamma =0 $ for $s \geq 0.2$}
    \label{fig:ent}
\end{figure}

We can then parametrize paths on the $s_1s_2$ plane that connect $H_0=H (0,0)$ with $H_1=H(1, 0)$. In particular, we choose \be
s_1 = s,\quad s_2 = \alpha s(1 - s).\label{s1s2parabola}\ee 
The respective paths for different values of $\alpha$ are shown in the inset of  Fig. \ref{fig:majorana}(b). We  calculate  the energy  gap  along these   paths, as  depicted in Fig. \ref{fig:majorana}(b). The gap does not close for any value of $s$ as long as  $\alpha \neq 0$, which means that there is a family of adiabatic paths that connect $\ket{\psi_0}$ and $\ket{\psi_1}$. \new{We emphasize that the adiabatic path constructed here is not unique, but in general it is important to break lattice symmetries to lift the distinction bewteen trivial and topological crystalline superconductors and avoid the gap closing in the free-fermion Hamiltonian \cite{Teo2013}.} We conclude that $\ket{\psi_0}$ and $\ket{\psi_1}$ are in the same phase, while $  P_G\ket{\psi_0}$ and $  P_G \ket{\psi_1}$ are not.  In the next subsection,  we perform  numerical calculations to find the critical values of $s$ where the transition occurs after the projection.

\subsection{Numerical results}

To distinguish  the two phases associated with  $  P_G \ket{\psi_0}$ ($\mathbb{Z}_2$ QSL) and $  P_G \ket{\psi_1}$ (trivial polarized state), we use the topological entanglement entropy \cite{kitaev2006topological,levin2006detecting}. The area law for the entanglement entropy of a partition with perimeter $L$ in two dimensions reads 
\begin{equation}
S(L) = a L - \gamma,
\end{equation}
where $a$ is a nonuniversal prefactor and  $\gamma = \ln \mathcal{D}$ is called the topological entanglement entropy. Here  $\mathcal D$  is the total quantum dimension, a universal property determined by the type of topological order. For a trivial phase, we have $\mathcal D = 1$. Meanwhile, for $\mathbb{Z}_2$ topological order it is known that $\mathcal D = 2$. Therefore, we expect to find the critical $s_c$ where $\gamma$ drops from $\ln 2$ to zero.

\begin{figure}[t]
    \centering
    \includegraphics[width=1.0\linewidth]{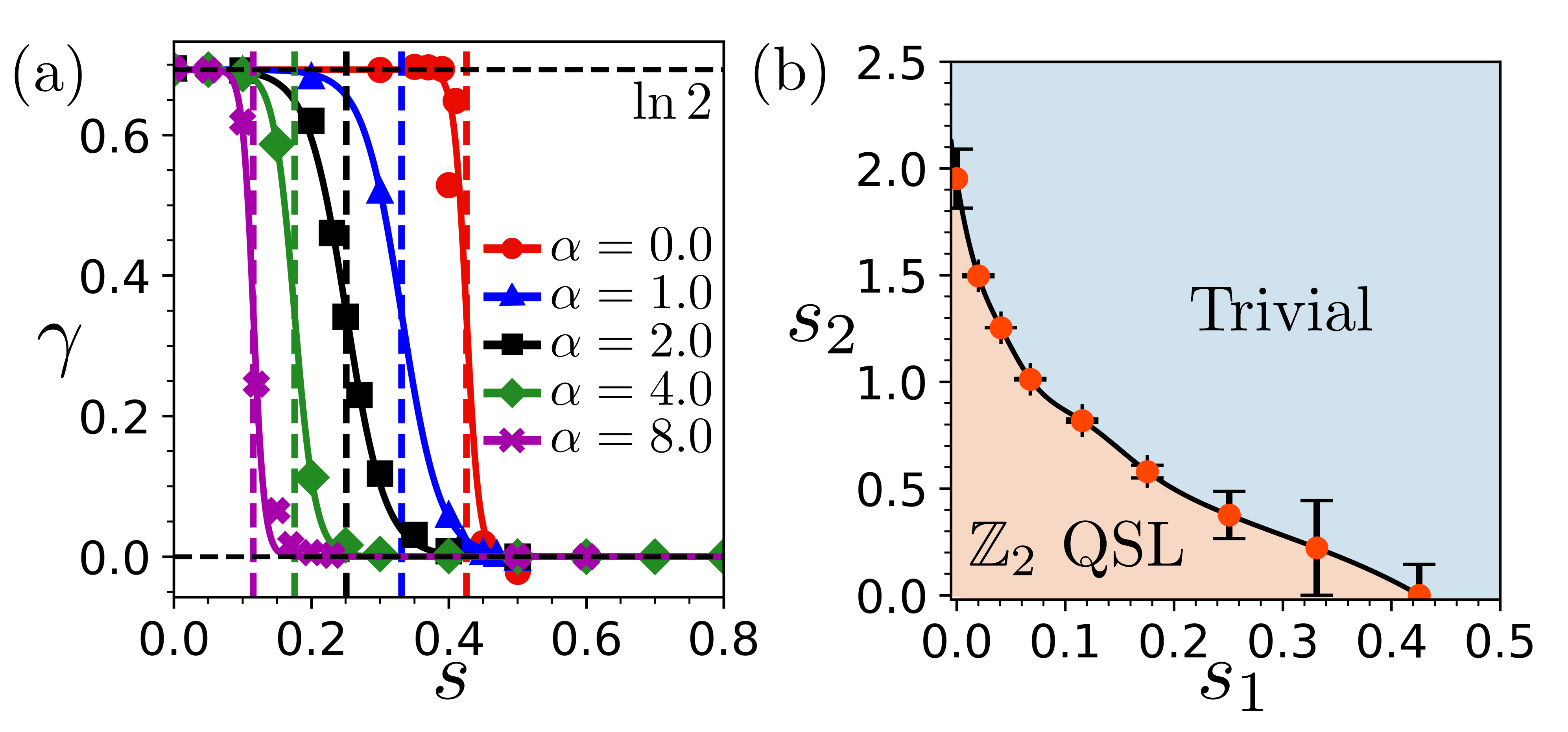}
    \caption{(a) Topological entanglement entropy as a function of $s$ for different adiabatic paths labeled by $\alpha$. The vertical dashed lines mark the inflection point of each curve, indicating the transition point.    (b) Phase diagram in the $s_1s_2$ plane,  see Eq.~(\ref{s1s2parabola}), for the projected states.}
    \label{fig:tee}
\end{figure}

To compute the topological entanglement entropy, we first compress   $\ket{\psi(s)}$ into an MPS following Ref. \cite{Jin2020}. Given the local form of an MPS, the Gutzwiller projection is naturally implemented in this language. For the MPS to accurately  represent  a gapped quantum state in two dimensions, its bond dimension $D$ has to grow exponentially with the linear system size.  With the MPS form of $ P_G \ket{\psi(s)}$, we determine  $\gamma$ following Ref. \cite{jiang2012identifying}. First, we impose open boundary conditions in one direction (say $x$) and periodic boundary conditions in the other ($y$), such that $N_x$ ($N_y$) is the number of sites in the open (periodic) direction. In all calculations, we use the bond dimension $D = 1024$, which is large  enough for $N_y \leq 8$.  Then, we find the entanglement entropy for a partition at $x = N_x/2$. For each $N_y$, we extrapolate the entanglement entropy for $N_x \rightarrow \infty$, as shown in Fig. \ref{fig:ent}(a). Last, we plot the results as a function of $N_y$ and fit them to $S(N_y) = aN_y - \gamma$. Figure \ref{fig:ent}(b) shows the results for different values of $s$ within the path with $\alpha = 4$. As expected, for all values of $s$, we get $\gamma = 0$ for $\ket{\psi(s)}$. Meanwhile, we get $\gamma = \ln 2$ for $P_G \ket{\psi(s)}$ if $s < s_c$ and $\gamma = 0$ for $s > s_c$. 

It is important to mention that, for this scheme to compute $\gamma$ to work, we need to ensure that $\ket{\psi(s)}$ is a minimally entangled state by carefully selecting the zero modes that could appear in the spectrum of $H(s)$ \cite{HHTu2020_Kitaev}. For $s = 0$, and using open boundary conditions in the $x$ direction, there are $2N_y$ zero modes. However, for $s > 0$, we observe that this degeneracy is lifted, making the choice of the occupied single-particle modes trivial. Given that we already had the analytical result for $P_G \ket{\psi_0}$, we choose to compute $P_G \ket{\psi(s)}$ for $s \geq 10^{-6}$   to avoid any ambiguity.

Using this algorithm, we find $\gamma$ along the adiabatic paths defined by different values of $\alpha$ in Eq. (\ref{s1s2parabola}). The results are shown in Fig. \ref{fig:tee}(a). While the numerical results obtained for  finite-size systems vary smoothly with the parameter $s$, we can estimate the critical value $s_c$ by taking  the inflection point of each curve that interpolates the numerical data. Our results show  that the   quantum phase transition  occurs at a finite   value $s_c \in (0,1)$. In particular, note that  for $\alpha=0$ the critical point  for the topological phase transition in the projected state  differs from the value of $s$ at which the gap of the free-fermion Hamiltonian vanishes. Having calculated   the values of $s_c(\alpha)$, we can draw the line that separates  the two phases in the $s_1s_2$ plane, parametrized by  $s_1 = s_c(\alpha)$ and $s_2 = \alpha s_c(\alpha) [1 - s_c(\alpha)]$. The corresponding phase diagram is represented in Fig. \ref{fig:tee}(b), \new{where the error bars represent the error  estimated from the width of the drop in the entanglement entropy.}

Using the projected MPS, we also compute local observables that are useful for characterizing the transition.  Figure \ref{fig:magplaq}(a) shows the magnetization \be
m_z(s)=\frac1{N_xN_y}\sum_{j}\braket{\psi(s)|P_G Z_jP_G|\psi(s)}\ee
for the path with $\alpha = 4$  and its respective  susceptibility  $\chi_z = d m_z/ds$. We can see that $\chi_z$ peaks close to the previously obtained critical point, with the peak position and height weakly dependent on $N_x$.  A similar trend is observed in $m_z$, as the curves for different sizes of the system overlap.  The same behavior is observed when the system size is increased in the periodic direction. \new{Here we  rely on the topological entanglement entropy as our main criterion to identify the critical point because this nonlocal quantity is more suited to capture a topological phase transition.} In Fig.  \ref{fig:magplaq}(b) we show  the expectation value of the plaquette operator $F_p=  Z_i X_jZ_kX_l$, which appears in the  Wen plaquette model in  Eq. (\ref{eq:wenhamiltonian}).  We have $\langle F_p\rangle=1$ at $s=0$ and $\langle F_p\rangle=0$ and $s=1$. We observe that the drop from one to zero along different adiabatic paths accompanies the previously computed value of the critical point. We note, however, that  the decay is much broader than the one obtained for the topological entanglement entropy, as expected because  the local observable  $\langle F_p\rangle$  is not an order parameter for the topological phase transition.

\begin{figure}[t]
    \centering
   \includegraphics[width=1.0\linewidth]{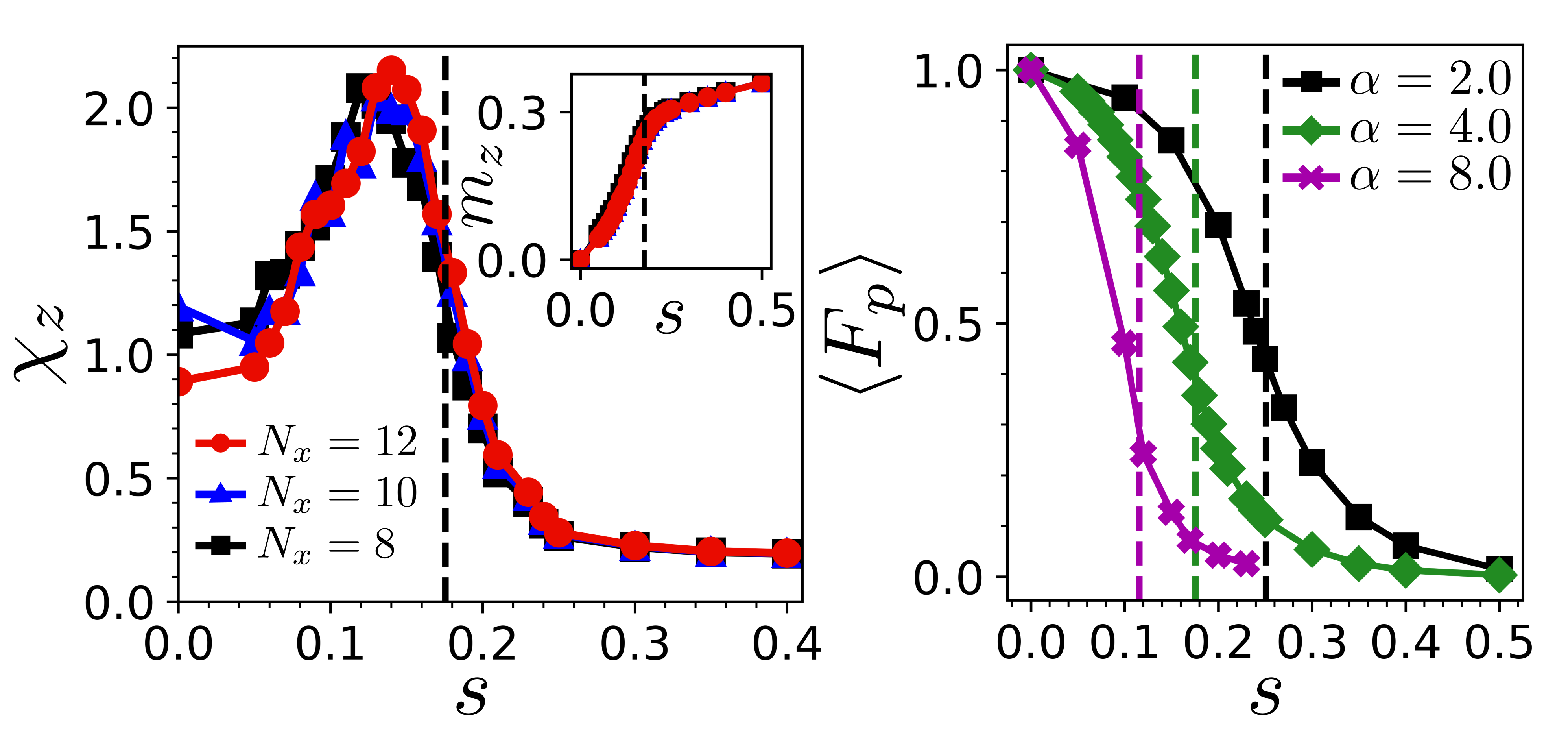}
    \caption{(a)  Magnetization $m_z$ (inset) and susceptibility $\chi_z$ along the adiabatic path with $\alpha = 4.0$ for different cylinder lengths $N_x$ and $N_y = 6$.  (b) Expectation value of the plaquette operator $\braket{F_p}$ for different adiabatic paths. The vertical dashed lines indicate the respective value of $s_c$  estimated from the entanglement entropy. }
    \label{fig:magplaq}
\end{figure}

Finally, we also evaluate the correlation matrix $C_{ij}^{zz}(s) = \braket{\psi(s)| {P}_G Z_i Z_j  {P}_G |\psi(s)}$. The result is shown in Fig. \ref{fig:corr}(a). We observe an exponential decay as a function of distance $r$ with a finite correlation length for all values of $s$. At first, this seems trivial, given that an MPS of finite bond dimension $D$ always has a finite correlation length \cite{Cirac2021}. However,  we verify that all observables converge for $D = 1024$ and $N_y \leq 8$, suggesting that the obtained MPS closely approximates the true projected state $  P_G \ket{\psi(s)}$. \new{For example, Fig. \ref{fig:corr}(b) shows the maximum value of the correlation length as a function of the bond dimension along the path with $\alpha = 4$ for fixed $N_x = 12$. As shown in Figs. \ref{fig:corr}(c) and (d), the maximum value of $\xi$ occurs near, but not exactly at  the value of $s_c$ obtained from the topological entanglement entropy. Importantly, we observe that the circumference $N_y$ dominates the finite-size effects in the correlation length. Given the limitations in the available system sizes, we are not able to conclusively establish whether the correlation length remains finite or diverges at the critical point in the thermodynamic limit. We note that in the analogous problem of a topological phase transition driven by a magnetic field in the perturbed toric code, both  first-order and continuous transitions are possible depending on the  field direction \cite{Dusuel2011}. } 

\begin{figure}
    \centering
    \includegraphics[width=1.0\linewidth]{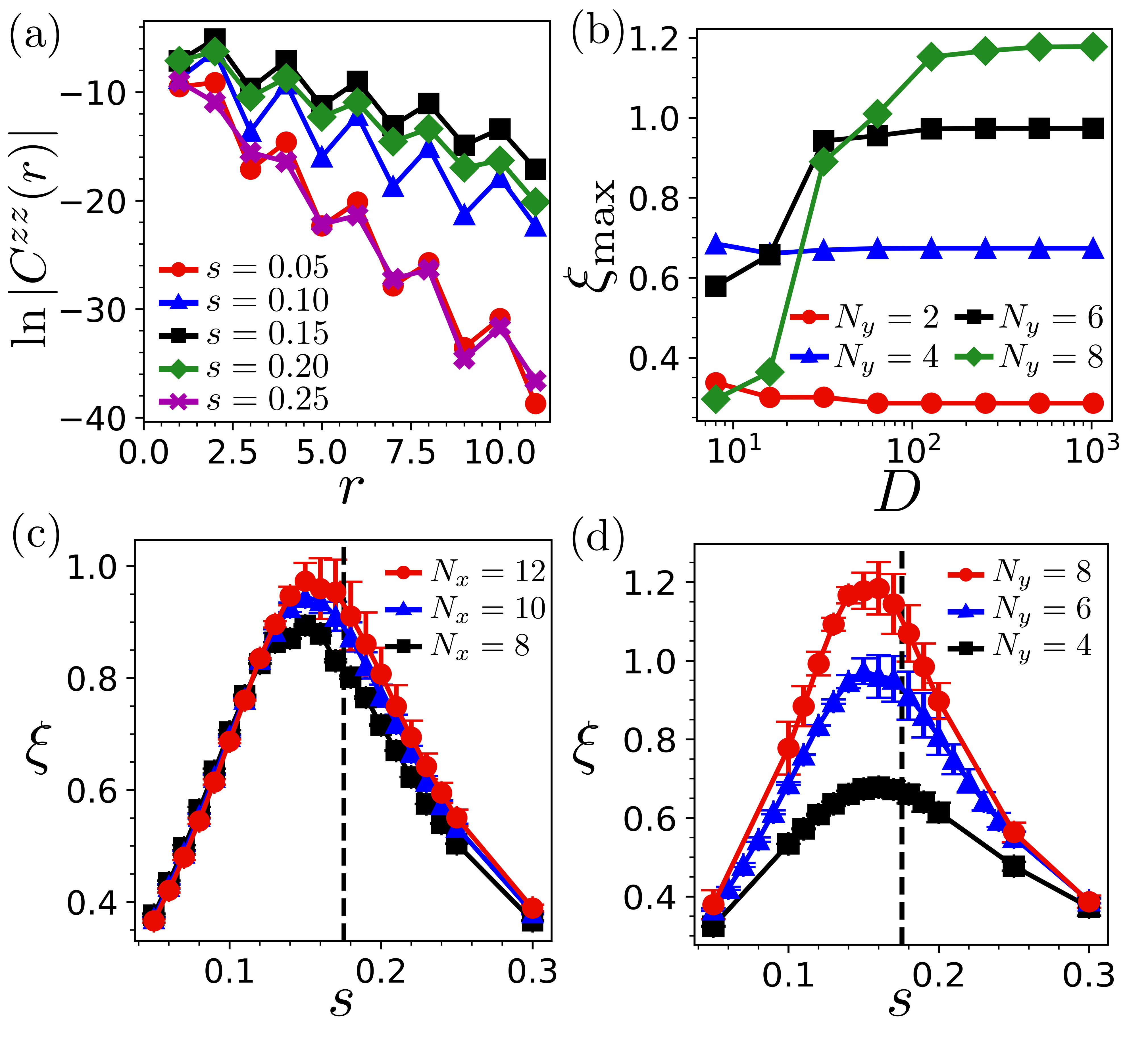}
    \caption{\new{Results for the evaluation of the correlation matrix using the projected MPS  along the adiabatic path with $\alpha = 4.0$. (a) Exponential decay of the correlation function as a function of distance for different values of $s$. (b) The maximum value that the correlation length $\xi$ governing the exponential decay in (a) assumes along the adiabatic path, $\xi_{\mathrm{max}}$, as a function of the bond dimension $D$ for different cylinder circumferences $N_y$ and fixed $N_x = 12$.  The value of $\xi$ along the adiabatic path for (c) fixed $N_y = 6$  and (d) fixed $N_x = 12$, varying the values of $N_y$ and $N_x$, respectively.  The vertical dashed line indicates the value of $s_c$  estimated from the entanglement entropy.}} 
    \label{fig:corr}
\end{figure}

 \section{Conclusions \label{sec:conclusion}}

  We have discussed the relation between Majorana fermion codes and the corresponding  Gutzwiller-projected spin states. We focused  on lattices that host  four Majorana fermions per site, which  can be expressed  in terms of spin-$1/2$ fermionic partons of quantum spin liquids. The unprojected fermionic state is the ground state of a mean-field Hamiltonian in which each Majorana fermion is   paired with another one in its vicinity. The key step in our approach   is to identify the spin stabilizer group  generated  by products of Majorana bilinears that commute with the  projector  imposing the local single-occupancy constraint. When this stabilizer group is two-dimensional and two-colorable, meaning that the projected graph tesselates the closed surface in a way that  two plaquettes labeled by the same color never share an edge, we demonstrate that the projected state exhibits all the characteristics of $\mathbb Z_2$ topological order. The construction also yields  trivial product states from the projection of fermionic states  in which all Majorana fermions are paired within the same site.

We also investigated the topological phase transition in projected fermionic states. For this purpose, we considered states that interpolate between two Majorana fermion codes which are both topologically trivial before the projection but become either a $\mathbb Z_2$ spin liquid or a trivial polarized state after the projection.  Our numerical results for the topological entanglement entropy allow us to pinpoint the critical  value of the  parameter that controls  the projected wave function.  While we have not referred to a specific spin Hamiltonian, this transition is analogous to the quantum phase transition driven by a magnetic field  in the perturbed toric code.

Our main conclusion is that caution is warranted when relying on the correspondence between free-fermion states and quantum spin liquid states. The Gutzwiller projection can introduce  significant, and sometimes unexpected,    effects. Therefore, one should check the topological properties of the projected states to ensure that they are in the same phase. This kind of  analysis could be extended to  parton constructions of other quantum spin liquids beyond  stabilizer states, for instance, chiral spin liquids  obtained via the projection of Chern insulators.

\begin{acknowledgments}
C.C.B. and E.C.A. acknowledge useful discussions with H-H Tu. We acknowledge funding by  Brazilian agencies CAPES (R.A.M.) and Conselho
Nacional de Desenvolvimento Cient\'ifico e Tecnol\'ogico  -- CNPq  (R.G.P. and E.C.A). C.C.B and E.C.A were supported by FAPESP, Grants No. 2021/06629-4, No. 2022/15453-0, and No. 2023/06101-5. This work was supported by a grant from the Simons
Foundation (Grant No. 1023171, R.G.P.), by  Finep (Grant No. 1699/24 IIF-FINEP, R.G.P.),  the MOST Young Scholar Fellowship (Grants
No. 112-2636-M-007-008- and No. 113-2636-M-007-002-), National Center for Theoretical Sciences (Grant No.
113-2124-M-002-003-) from the Ministry of Science and
Technology (MOST), Taiwan, and the Yushan Young
Scholar Program (as Administrative Support Grant Fellow) from the Ministry of Education, Taiwan. Research at IIP-UFRN is supported by Brazilian ministries MEC and MCTI.
\end{acknowledgments}

\section*{Data Availability}

The data that support the findings of this article are openly available \cite{data}.

\appendix

\section{Proof of Lemma 1\label{app:proof}}

First, we establish the isomorphism. In the following we   refer to the stabilizer group $\mathcal D = \langle \{\gamma^0_j \gamma^1_j \gamma^2_j \gamma^3_j\}_{j=1}^n\rangle $ and the corresponding centralizer $\mathrm{C}_{\mathrm{Maj}(4n)}(\mathcal D)$.  Any operator  $m \in \mathrm{C}_{\mathrm{Maj}(4n)}(\mathcal D)$ must contain two Majorana modes per site since this is the condition to commute with all local parities in  $\mathcal D$. That is, the group admits the following presentation:
\be
\mathrm{C}_{\mathrm{Maj}(4n)}(\mathcal M ) = \langle iI, \{i \gamma_i^{p_{i,1}} \gamma_i^{p_{i,2}} \}_{p_i \in {[4] \choose 2} \;, i=1} ^n \rangle\;,
\ee
where $p_i = (p_{i,1},p_{i,2}) \in {[4]\choose 2}$ is one of the $ {4 \choose 2} = 6$ pairings of $\{0,1,2,3\}$. 

The Kitaev's parton map on Eq. (\ref{eq:kitaevpartons}), define a group homomorphism ${\mathrm P}: \mathrm C_{\mathrm{
Maj
}(4n)}(\mathcal D) \to \mathcal P_n$ by its local action on the generators of the centralizer, as listed in Table \ref{eq:kitaevpartontable}.
\begin{table}[b!]
    \centering
    \begin{tabular}{c|c}
        $m$ & ${\mathrm{P}}(m)$ \\
        \hline\hline
        $i\gamma^1_j \gamma^0_j$ & $X_j$\\
        \hline 
        $i\gamma^2_j \gamma^3_j$ & $X_j $\\
        \hline 
        $i \gamma^2_j \gamma^0_j$ & $Y_j$\\
        \hline
        $i \gamma^3_j \gamma^1_j $ & $Y_j$ \\
        \hline
        $i \gamma^0_j \gamma^3_j $ & $Z_j$ \\
        \hline
        $i \gamma^2_j \gamma^1_j$ & $Z_j$ 
    \end{tabular}
    \caption{Action of Kitaev's parton map on Majorana bilinears on  site $j$.}
    \label{eq:kitaevpartontable}
\end{table}
By its action, it is clear that
\begin{equation}
    \mathrm{ker}(\mathrm P) = \mathcal D\;.
\end{equation}
Note that:
\be
\mathcal M \cap \mathcal D = \left\{1, \prod_i D_i\right\}\;,
\label{eq:intesectionmaj}
\ee
since we assume $\mathcal M$ to be parity-even. Now, we are  interested in constructing a stabilizer group for $P_G|\psi_\mathcal M\rangle$. First, consider the centralizer $\mathrm{C}_{\mathcal M}(\mathcal D) \subseteq \mathcal M$, defined as the subgroup of $\mathcal M$ that commutes with the local parities. 

Any element of $\mathcal M$ can be seen as a subset of the edges of $G_\mathcal M = (V,E)$, since every Majorana bilinear uniquely corresponds to an edge. To have $m \in \mathrm{C}_\mathcal M(\mathcal D)$,  the corresponding subgraph $G_\mathcal M[m]$, defined as the induced subgraph defined by its edges, must have even degree on all vertices. Hence, $G_\mathcal M[m]$ is a cycle, and thus, the minimal number of generators of the centralizer, $|\mathrm{gen} \mathrm C_\mathcal M(\mathcal D)|$, is the number of independent cycles on a graph. Since it has only one simply connected component, the corresponding formula is:
\begin{equation}
    |\mathrm{gen} \mathrm C_\mathcal M(\mathcal D)| = |E| - |V| +1 =n+1 \;,
    \label{eq:cycleformula}
\end{equation}
where we have used the fact that $|E| =2n$, due to $G_\mathcal M$ having only degree-four vertices.
Now, consider the group $\mathcal M_\mathcal D \equiv \langle \mathrm C_\mathcal M(\mathcal D), \mathcal D\rangle \subseteq \mathrm{Maj}^+(4n)$, which is a Majorana stabilizer group, with $P_G|\psi_\mathcal M \rangle \in V_{\mathcal M_\mathcal D}$. We know that $|\mathrm{gen}(\mathcal D)| = n$. Furthermore, since the overlap of $\mathcal M$ with $\mathcal D$ is  nontrivial with the global parity only, the number of generators of $\mathcal M_\mathcal D$ is
\be
|\mathrm{gen}(\mathcal M_\mathcal D)| = (n+1)+ n -1 = 2n\;,
\ee
meaning that $P_G|\psi_\mathcal M\rangle$ is the \emph{unique} state stabilized by $\mathcal M_\mathcal D$ since the corresponding fermionic encoding rate vanishes. We claim that the desired stabilizer group is:
\begin{equation}
    P\mathcal M \equiv \mathrm{im}_\mathrm{P}[\mathcal M_\mathcal D] =  \mathrm{im}_\mathrm P[\mathrm C_\mathcal M(\mathcal D)]\;,
\end{equation}
where $\mathrm{im}_\mathrm P(\cdot)$ denotes the image under $\mathrm P$.
It satisfies all the properties of a stabilizer state, since $\mathrm P$ is a homomorphism. Moreover,  $ P\mathcal M$ stabilizes a unique state, $|P\psi_\mathcal M\rangle$, since $|\mathrm{gen}(P\mathcal M)| =(n+1) -1$, given that   only $\prod_i D_i$ is in the kernel of $\mathrm P$.

We define the \emph{charge} of two operators $m_1,m_2 \in \mathrm{Maj}(4n)$ or $P_1, P_2 \in \mathcal P_n$ as:
\begin{align}
    m_1 m_2 &= q(m_1,m_2) m_2 m_1\;,\\
    P_1 P_2 &= q(P_1, P_2) P_2 P_1\;,
\end{align}
assigning a $\{-1,1\}\in \mathbb Z_2$ phase for each pair in the corresponding groups. It follows   that, given $m_1, m_2 \in \mathrm{C}_{\mathrm{Maj}(4n)}(\mathcal D)$, $q(m_1, m_2) = q(\mathrm P(m_1), \mathrm P(m_2))$. 
Now, we use the following fact: Given $\mathcal S$ as a qubit (or Majorana) stabilizer group that stabilizes a normalized, unique state $|\mathcal S\rangle$, if $O$ is an element of $\mathcal P_n$ (or $\mathrm{Maj}(4n)$), it follows that
\be
\langle \mathcal S| O |\mathcal S\rangle = \mathbf 1_\mathcal S(O)\;, 
\ee
where $\mathbf 1_\mathcal S$ is the indicator function of $O \in \mathcal S$ ($\mathbf 1_\mathcal S(O) =1$ if $O \in \mathcal S$, $\mathbf 1_\mathcal S(O) =0$ otherwise). This follows from the fact that, if the code has $k=0$, then $\mathrm C_{\mathcal P_n}(\mathcal S) \cong  \mathcal S$. Hence, either an operator $O$ is in the stabilizer group, with $\langle \mathcal S |O|\mathcal S\rangle =1$, or the Pauli anticommutes with some of the generators of the stabilizer group, implying a vanishing expectation value: $\langle \mathcal S|O|\mathcal S\rangle = 0$.

Then, consider $m \in \mathcal M_\mathcal D$.
It follows that $\mathrm P(m) \in 
P\mathcal M $, even if $m \in \mathcal D$, since in this case $\mathrm P(m) = 1$. If $m \not \in \mathcal M_D$, then $m \notin \mathrm{C}_{\mathrm{Maj}(4n)}(\mathcal M_\mathcal D)$, since $V_{\mathcal M_\mathcal D}$ is a $k=0$ code, and therefore, $\mathrm P(m) = 0$. We conclude, then:
\be
\frac{\langle \psi_\mathcal M|P_G mP_G|\psi_\mathcal M \rangle}{\langle \psi_\mathcal M|P_G|\psi_\mathcal M\rangle} = \langle P \psi_\mathcal M|\mathrm{P}(m) |P \psi_\mathcal M\rangle\;.
\ee

\section{ Proof of Lemma 2 \label{prooflemma2}}
To  tackle Lemma 2, let us briefly review some concepts from algebraic topology which will be useful for this result and all the following Lemmas.

First, let us define a chain complex. The idea starts by defining \emph{binary} vector spaces (vector spaces with the field $\mathbb F_2$) corresponding to the embedding of the projected graph into the surface. This is constructed from the embedding $\sigma: G_\mathcal M\to \Sigma$: Given $V, E$ as the vertices and edges of $G_\mathcal M$, the fact that $\sigma$ is cellular, implies the existence of a set of faces, $F_\sigma$. Then, the span over the binary field of vertices, edges, and faces defines the 0-,1-, and 2-chains, respectively:
\be
X_0 \equiv \mathrm{span}_{\mathbb F_2} V \, ; \quad X_1 \equiv \mathrm{span}_{\mathbb F_2} E \, ; \quad X_2 \equiv \mathrm{span}_{\mathbb F_2}F_\sigma\;.
\ee
The addition in the vector space is defined $\mod 2$. For example, for $v \in V$, we have $v+v = 0 \in X_0$. They come with a boundary structure. For example, given a facet $f \in F_\sigma$, there is a set $\partial f\subseteq E$ of edges that contains the boundary of $f$. Hence, we can define a linear map $\partial_2: X_2 \to X_1$ by its action on the faces:
\be
\partial_2 f = \sum_{e \in \partial f}e\;,
\ee
and doing linear extension. Similarly, the linear map $\partial_1: X_1 \to X_0$ is   defined by taking the vertices where each edge is defined, that is, if $e = (v_0, v_1)$, then $\partial_1e = v_0 + v_1$. A chain complex is defined by the sequence of chains and its boundary maps:
\be
X_\sigma:0 \xrightarrow{\partial_3} X_2 \xrightarrow{\partial_2}X_1 \xrightarrow{\partial_1}X_0 \xrightarrow{\partial_0}0\;,
\ee
where we added the trivial vector space $0$ in both ends, with the trivial inclusion $\partial_3(0) = 0 \in X_2$ and projection $\partial_0(v \in X_0) = 0$, for reasons that will become clear shortly, it is useful to have both in and out boundary maps in each nontrivial chain. Note that the boundary maps satisfy the consistency condition
\begin{equation}
    \partial_i \partial_{i+1} = 0\;.
    \label{eq:nilpotencybound}
\end{equation}
Given $i \in \{0,1,2,\}$, this motivates the definition of the $i$-cycles $Z_i$ and $i$-boundaries $B_i$ subspaces as 
\begin{align}
&Z_i(X_\sigma) \equiv \mathrm{ker}\partial_i = \{z \in X_i \;|\; \partial_i z = 0\}\;,\\
&B_i(X_\sigma) \equiv \mathrm{im} \partial_{i+1} = \{b \in X_i \;|\;\exists x \in X_{i+1}:\; \partial_{i+1}x=b\}\;.
\end{align}
As their name suggests, cycle spaces describes boundaryless chains, and boundary spaces corresponds to chains which are boundaries of another. Due to the nilpotency condition in Eq. (\ref{eq:nilpotencybound}), it follows that $B_i(X_\sigma) \subseteq Z_i(X_\sigma)$. Their ``difference'' is captured by homologies, defined through the quotient
\be
H_i(X_\sigma) = Z_i(X_\sigma)/B_i(X_\sigma)\quad ; \quad i \in \{0,1,2\}\;.
\ee
The idea is that elements of the homology captures topological invariants of dimension $i$ of the target manifold $\Sigma$. For example, if the graph is to be embedded on $\Sigma = S^2$, then every 1-cycle can be obtained as a boundary of some region, as a consequence of being contractible, and hence $H_1(X_\sigma) = 0$. However, if the graph is embedded on a torus,   $\Sigma = T^2$, there are nontrivial 1-cycles which wrap around its handles, making $H_1(X_\sigma) \neq 0$. In fact, the target manifold itself holds homology spaces $H_i(\Sigma, \mathbb F_2)$ (read as the $i$-th homology space/group with $\mathbb F_2$ coefficients), isomorphic to any chain complex defined through a cellular embedding $\eta$ \cite{hatcher2001algebraic}:
\be
H_i(X_\eta) \cong H_i(\Sigma, \mathbb F_2)\;.
\ee
Correspondingly, one can also consider the dual chain complex $X_\sigma^*$, obtained from $X_\sigma$ by dualizing the vector spaces and maps. Referring to $d_{3-i} \equiv \partial_i^T$ as the coboundaries, we obtain
\be
X^*_\sigma:0 \xleftarrow{d_0} X^*_2 \xleftarrow{d_1}X^*_1 \xleftarrow{d_2}X^*_0 \xleftarrow{d_3}0\;.
\ee
Note that $X_i^* \cong X_i$. Similarly,  given  $i \in \{0,1,2\}$, we can define  cocycles and coboundaries as
\be
Z^i(X_\sigma) \equiv \ker d_{i} \quad ; \quad B^i(X_\sigma) \equiv \mathrm{im} d_{i+1}\;.
\ee
They are not completely independent from the cycles and boundaries. Denoting $\cdot: \mathbb F_2^n\times \mathbb F_2^n \to \mathbb F_2$ as the symmetric inner product on which the duality isomorphism is defined, we have $x \cdot \partial_i y = d_{n-i}x\cdot y$. Then, for any chain complex, the set of spaces $\{Z_i, B_i, Z^i, B^i\}$ are related by:
\be
Z_i^\perp \cong B^i \;,\; [Z^i]^\perp\cong B_i\;.
\label{eq:orthocyclesbound}
\ee
The corresponding cohomology spaces are defined as:
\be 
H^i(X_\sigma)  \equiv Z^i(X_\sigma)/B^i (X_\sigma)\;.
\ee
One can interpret the duals as the chain complex derived from the embedding of the dual of $G_\mathcal M$, where the vertices corresponds to the faces $F_\sigma$, and an edge is defined on a pair of faces if they share an edge on the original embedding. This motivates the result known as the Poincaré Lemma:
\be
H^{n-i}(X_\sigma) \cong H_i(X_\sigma)\;.
\ee

We will also need  \emph{combinatorial systoles}. Let $x \in X_i$, and define its support as the subset of vertices, edges or faces where it can be written as a decomposition, $x  = \sum_{s \in \mathrm{supp}(x)} s$. Then, its weight is defined as
\be
\mathrm{wt}(x)\equiv |\mathrm{supp}(x)|\;,
\label{eq:weightformula}
\ee
and equivalently for the duals, by the isomorphism $X_i^* \cong X_i$. The combinatorial systoles are defined as
\begin{align}
    \mathrm{csys}_i(X_\sigma)&\equiv \min_{\substack{z \in Z_i(X_\sigma)\\ [z] \neq 1 \in H_i(X_\sigma)}} \mathrm{wt}(z)\;,\\
    \mathrm{csys}^*_i(X_\sigma)&\equiv \min_{\substack{z \in Z^i(X_\sigma)\\ [z] \neq 1 \in H^i(X_\sigma)}} \mathrm{wt}(z)\;.
    \label{eq:systoles}
\end{align}
They are interpreted as the minimum length of nontrivial cycles/cocycles in the chain complex. They are of combinatorial nature because  they explicitly depend on the graph embedding.

We can now present a Proposition connecting our discussion of Majorana codes and homological notions:

\begin{proposition}
Let $(\mathcal M, \Sigma)$ be a two-dimensional Majorana group. Then, the corresponding centralizer over $\mathcal D$ is isomorphic to 1-cycles, that is $\mathrm C_\mathcal M(\mathcal D) \cong Z_1(X_\sigma)$.
\end{proposition}

\begin{proof}
First, note that one can reduce each of the chains $X_i$ as groups, where the vector space addition is seen as the group operation. We then have  the isomorphisms
\be
X_0 \cong \mathbb Z_2^{|V|} \quad  ; \quad X_1 \cong \mathbb Z_2^{|E|} \quad ; \quad X_2 \cong \mathbb Z_2^{|F_\sigma|}\;.
\label{eq:spacetogroup}
\ee
Since $\mathcal M$ is generated by the Majorana dimers that define the edges, it also follows that $\mathcal M \cong \mathbb Z_2^{|E|} $, implying that $\mathcal M$ and $X_1$ are isomorphic. This was already indicated  by the action on the Majorana bilinears in the map from  which  $E$ was constructed, since   $i \gamma^a_i \gamma^b_j \mapsto e^{ab}_{ij}\equiv \{(i,j), (a,b)\} \in X_1$. We   denote the isomorphism by $S: \mathcal M\to X_1$. We then note:
\begin{itemize}
    \item Let $m \in \mathrm{C}_\mathcal M(\mathcal D)$, which must have an even number of Majoranas per site   to commute with the local parities. Hence, the corresponding element in the 1-chain $S(m) \in X_1$ must have a vanishing boundary, $\partial_1 S(m) = 0$, showing that $S(m) \in Z_1(X_\sigma)$.
    \item Consider now some 1-cycle $z \in Z_1(X_\sigma)$. Note that the corresponding Majorana string $S^{-1}(m) \in \mathcal M$ must commute with the local parities, since the boundaries of each edge in $z$ have even degree in the induced subgraph, proving  that $S^{-1}(m) \in \mathrm{C}_\mathcal M(\mathcal D)$ and thus $S$ induces $\mathrm C_\mathcal M(\mathcal D) \cong Z_1(X_\sigma)$.
\end{itemize}
\end{proof}

Note that this isomoprhism was already foreshadowed in the proof of Lemma 1: The counting in Eq. (\ref{eq:cycleformula}) is nothing more than the dimension of the 1-cycles, $\mathrm{dim}Z_1(X_\sigma) = n+1$, mapping to the number of group generators under Eq. (\ref{eq:spacetogroup}).

The strategy now is to use Proposition  1 to map the problem onto a cycle decomposition problem, and then see how it behaves under the various maps until $P \mathcal M$ is reached.

Consider $Z_1(X_\sigma) = \ker \partial_1$. It decomposes as a union:
\be
Z_1(X_\sigma)= B_1(X_\sigma) \cup Z_1^h(X_\sigma)\;,
\ee
where $Z_1^h(X_\sigma) = Z_1(X_\sigma) \backslash B_1(X_\sigma)$ is the complement of the cycles with respect to boundaries, which corresponds to the 
essential cycles. We note that they do not close under the vector space addition, since a pair of nontrivial cycles can yield a trivial one in homology. By the isomorphism shown in Proposition 1, a similar decomposition happens in the centralizer:
\be
\mathrm C_\mathcal M(\mathcal D) \cong \mathrm C_\mathcal M^\partial( \mathcal D) \cup \mathrm C_\mathcal M^h(\mathcal D)\;,
\ee
where $\mathrm C_\mathcal M ^\partial(\mathcal D) \cong B_1(X_\sigma)$ and $\mathrm C^h_\mathcal M(\mathcal D)$ is the isomorphism image of $Z_1^h(X_\sigma)$. Then, the claimed decomposition $P\mathcal M  = P_\partial \mathcal M \cup P_h \mathcal M$ comes from the definition:
\begin{align}
    P_\partial \mathcal M &\equiv \mathrm{im}_\mathrm P[\mathrm C_\mathcal M^\partial(\mathcal D)]\;,\\
    P_h \mathcal M &\equiv \mathrm{im}_\mathrm P[\mathrm C_\mathcal M^h(\mathcal D)]\;.
\end{align}
It follows that $P_\partial \mathcal M, P_h\mathcal M$ have support at most as the weight [see Eq. (\ref{eq:weightformula})] of trivial and essential cycles, since the action of $\mathrm P$ can only reduce the weight.

%\section{CSS structure of projected Majorana codes}

\section{Proof of Theorem 1\label{proofT1}}

In this section, we  prove two propositions that fix the algebraic structure of projected codes and are important to prove the theorem stated in the main text.

First, we need a definition of a classical linear code, which is  simply  a   collection of bit strings of $n$ bits, that is, $C \subseteq \mathbb F_2^n$, with $ \mathrm{dim}C=2^k$. The natural errors we want to check are bit flips. It is useful to define the weight $\mathrm{wt}(x)$ of a bit string $x \in \mathbb F_2^n$  as the number of $1$'s in the bit string. Then, the natural distance measure is
\be
d= \min_{x,y \in C} \mathrm{wt}(x-y) =\min_{z \in C}\mathrm{wt}(z)\;.
\ee
We refer to $C$ as a $[n,k,d]$ classical code. A natural way to define a code is to specify its parity-check matrix $h: \mathbb F_2^n \to \mathbb F_2^{n-k}$, such that $C = \ker h$. The elements of the code, $x \in C$, referred to as codewords, are annihilated by $h$, $h\cdot x  =0$.

A simple example is to consider $n=3$, and take $C = \mathrm{span}_{\mathbb F_2}\{000,111\} \in \mathbb F_2^3$, referred to as the repetition code. It encodes $k=1$ bit and has distance $d = \mathrm{wt}(111)=3$, forming a $[3,1,3]$ code. The corresponding $2\times 3$ parity-check matrix $h_\mathrm{rep}$ is specified by the linear equations obtained from $h_\mathrm{rep}\cdot x=0$. For $x = x_1x_2x_3$, it would be $x_1+x_3=0$ and $x_2+x_3 = 0$. Their form motivates the name of $h$: Given some bit string $x$, each equation defined by the rows of $h$ defines a parity-check equation on the bit string, which confirms its presence in the code upon vanishing of all of them.

We will now provide a description to define special qubit stabilizer codes from a pair of classical ones. It has two ingredients:
\begin{itemize}
    \item A pair of classical linear codes $C_1,C_2 \subseteq \mathbb F_2^n$, with datum $[n,k_1,d_1]$ and $[n, k_2, d_2]$. Given the corresponding parity-check matrices $h_{1,2}:\mathbb F_2^n \to \mathbb F_2^{r_{1,2}}$, we impose that the orthogonal complements are contained in each other, that  is, $C_1^\perp \subseteq C_2$ and $C_2^\perp \subseteq C_1$. Thus, their parity-check matrices satisfies $h_1 h_2^T = h_2h_1^T=0$.
    \item a pair of Pauli tuples, $\mathbf P_{1}, \mathbf P_2 \in \mathcal P_1^{\times n}$, with $P_{i,k}$ supported on qubit $k$. We further impose them to be anticommuting, that is, $P_{1,k}P_{2,k} = -P_{2,k}P_{1,k}$. 
\end{itemize}
Note that given any Pauli tuple $\mathbf P \in \mathcal P_1^{\times n}$, we can pair it with a bit string $x \in \mathbb F_2^n$, defining 
\begin{equation}
    \mathbf P^x\equiv \bigotimes_{i=1}^n P_i^{x_i}\;.
\end{equation}
Then, since the rowspace of both parity-check matrices can be represented as $r_1= n-k_1$, $r_2=n-k_2$ bistrings of $n$ elements, denoted as $h_{i,1}, \cdots, h_{i, r_i}$, we can define the following stabilizer groups:
\begin{align}
\mathcal S_{\mathbf P_1} (C_1) &\equiv \langle \mathbf P_1^{h_{1,1}},\cdots, \mathbf P_1^{h_{1,r_1}}\rangle\;,
\label{eq:css1}\\
\mathcal S_{\mathbf P_2} (C_2) &\equiv \langle \mathbf P_2^{h_{2,1}},\cdots, \mathbf P_2^{h_{2,r_2}}\rangle\;.
\label{eq:css2}
\end{align}
A nontrivial fact is that they commute with each other, $[\mathcal S_{\mathbf P_1}(C_1), \mathcal S_{\mathbf P_2}(C_2)]=0$, due to the orthogonality condition on the codes. Given any $1 \leq m \leq r_1$ and $ 1 \leq n \leq r_2$:
\be
\mathbf P_1^{h_{1,m}} \mathbf P_2^{h_{2,n}} = (-1)^{(h_1 h_2^T)_{mn}} \mathbf P_2^{h_{2,n}} \mathbf P_1^{h_{1,n}}\;,
\ee
where $(h_1 h_2^T)_{mn}$ denotes the $(m,n)$ element of $h_1 h_2^T$. Since it vanishes due to the orthogonality condition, the corresponding groups commute. Our desired group is defined as
\be
\mathcal S_{\mathbf P_1, \mathbf P_2}(C_1, C_2) \equiv \langle \mathcal S_{\mathbf P_1}(C_1), \mathcal S_{\mathbf P_2}(C_2)\rangle\;,
\ee
which is stabilizer exactly due to the commutativity of the generating groups.

It turns out that the properties of these codes can be tightly encoded into a chain complex, defined as:
\begin{equation}
X(C_1, C_2): \mathbb F_2^{r_2}  \xrightarrow{h_2^T} \mathbb F_2^{n} \xrightarrow{h_1} \mathbb F_2^{r_1}\;,
\end{equation}
satisfying Eq. (\ref{eq:nilpotencybound}) by definition. Then, we have:

\noindent 
\begin{proposition}
Let $C_1, C_2 \subseteq \mathbb F_2^n$ be two linear codes with parameters $[n, k_1, d_1]$, $[n, k_2, d_2]$ with $C_1^\perp \subseteq C_2$, and let $\mathbf P_1, \mathbf P_2 \in \mathcal P_1^{\times n}$ be two anticommuting Pauli tuples. Then, $\mathcal S_{\mathbf P_1, \mathbf P_2}(C_1, C_2) \cong \mathcal S_{\mathbf P_1}(C_1) \times \mathcal S_{\mathbf P_2}(C_2)$ is a $[[n,k,d]]$ qubit stabilizer code with:
\begin{enumerate}
    \item $k = \mathrm{dim} H_1[X(C_1, C_2)] = \mathrm{dim} H^1[X(C_1, C_2)] = n-k_1-k_2$,
    \item $d = \min(\mathrm{csys}_1[X(C_1, C_2)], \mathrm{csys}^*_1[X(C_1,C_2)]) \geq \min(d_1,d_2)$.
\end{enumerate}
\end{proposition}

\begin{proof}
 First, note that the classical codes can be written as the cycle and cocyle spaces:
\begin{equation}
    C_1 =  Z_1[X(C_1,C_2)] \quad , \quad C_2 =  Z^1[X(C_1, C_2)]\;,
    \label{eq:codecycles}
\end{equation}
since they are kernels of the parity-check matrices, which are the respective boundary and coboundary operators of $X(C_1, C_2)$. Furthermore, due to Eq. (\ref{eq:orthocyclesbound}), we get
\begin{equation}
C_1^\perp = B^1[X(C_1,C_2)] \quad , \quad C_2^\perp = B_1[X(C_1,C_2)]\;.
\label{eq:orthocodeboundaries}
\end{equation}
We can now show both results.

\noindent
1. $k$ can be computed as
\begin{equation}
    k = n-r_1-r_2\;,
    \label{eq:encodingCSS}
\end{equation}
where we have used the fact that $\mathcal S_{\mathbf P_1}(C_1) \cap \mathcal S_{\mathbf P_2}(C_2) = \varnothing$. Since we have $r_1 = n-k_1$ and $r_2 = n-k_2$, the formula $k = n-k_1-k_2$ is shown. To show that it is related to the dimension of the first homology space, known as the Betti number, first note that:
\begin{align}
    r_1 &= \mathrm{dim}C_1^\perp = \mathrm{dim} B_1[X(C_1, C_2)]\;,\\
    r_2 &= \mathrm{dim}C_2^\perp = \mathrm{dim} B^1[X(C_1, C_2)]\;,
\end{align}
relating the encoding rate to the dimensions of the boundaries and coboundaries, $k=n-\mathrm{dim} B_1[X(C_1, C_2)] - \mathrm{dim} B^1[X(C_1, C_2)]$.
Given two vector spaces $V, W$ and a linear map $V \xrightarrow{f} W$, we know that $W \cong \mathrm{im} f \oplus \ker f$, relating the dimensions as $\mathrm{dim}W = \mathrm{dim}(\mathrm{im} f) + \mathrm{dim}(\ker f)$. Hence, by invoking the relations in Eq. (\ref{eq:orthocyclesbound}) again, the actions $\mathbb F_2^{r_2} \xrightarrow{h_2^T} \mathbb F_2^n$ and $\mathbb F_2^{r_1} \xrightarrow{h_1^T} \mathbb F_2^n$ give the corresponding relations:
\begin{align}
n &= \mathrm{dim} B_1[X(C_1,C_2)] + \mathrm{dim} Z^1[X(C_1,C_2)]\;,\\
n &= \mathrm{dim} B^1[X(C_1,C_2)] + \mathrm{dim} Z_1[X(C_1,C_2)]\;.
\end{align}
Using the first relation, we obtain $k=n-\mathrm{dim} B^1[X(C_1, C_2)] - (n-\mathrm{dim} Z^1[X(C_1, C_2)]) = \mathrm{dim} H^1[X(C_1, C_2)]$, using the fact that the dimension of a quotient space $V/W$ is the difference $\mathrm{dim} V- \mathrm{dim} W$. The second relation yields $k= \mathrm{dim}H_1[X(C_1, C_2)]$, proving the formula.

\noindent 
2. Due to the anticommuting property of the Pauli tuples, note that one can write the $n$ qubit Pauli group as $\mathcal P_n = \langle i, \mathbf P_1, \mathbf P_2\rangle$. That is, given $p \in \mathcal P_n$, there  exist $a,b \in \mathbb F_2^n$ such that:
\begin{equation}
    p \propto \mathbf P_1^a \mathbf P_2^b\;.
\end{equation}
The strategy to show the distance bound is to impose the constraint that $p$ must satisfy, and bound the weights of $a$ and $b$.

Note that imposing $p \in \mathcal S_{\mathbf P_1, \mathbf P_2}(C_1, C_2)$ is equivalent to imposing that $a \in C_1^\perp$ and $b \in C_2^\perp$, since the latter are the vector spaces spanned by the parity-check rows. If we relax   the condition to impose $p \in \mathrm{C}_{\mathcal P_n}[\mathcal S_{\mathbf P_1, \mathbf P_2}(C_1, C_2)]$, then this would correspond to $b \in C_1$ and $a \in C_2$, since the commutativity constraints to Eqs. (\ref{eq:css1}, \ref{eq:css2}) impose that the two bit strings be code words. Hence, $p$ is a logical Pauli of $S_{\mathbf P_1, \mathbf P_2}(C_1, C_2)$ if and only if
\be
a \in C_1 \backslash C_2^\perp \quad ; \quad b \in C_2 \backslash C_1^\perp\;.
\ee
Noting that $|\mathrm{supp}\mathbf P_1^\mathbf a| = \mathrm{wt}(\mathbf a) \leq |\mathrm{supp}[\mathbf P_1^\mathbf a \mathbf P_2^\mathbf b]$ and $|\mathrm{supp}(\mathbf P_2^\mathbf b| =\mathrm{wt}(\mathbf b) \leq |\mathrm{supp}[\mathbf P_1^\mathbf a (\mathbf P_2)^\mathbf b]$, we find that the distance is given by
\begin{align}
    d  &= \min_{l \in \mathcal L[\mathcal S_{\mathbf P_1, \mathbf P_2}(C_1, C_2)]}|\mathrm{supp}(l)|\\ 
    &= \min\left[\min_{\mathbf a \in C_1\backslash C^\perp_2} \mathrm{wt}(\mathbf a),\min_{\mathbf b \in C_2\backslash C^\perp_1} \mathrm{wt}(\mathbf b) \right]\;. 
    \label{eq:distancecodes}
\end{align}
Due to the identifications in Eqs. (\ref{eq:codecycles}, \ref{eq:orthocodeboundaries}) and the definition of combinatorial systoles in Eq. (\ref{eq:systoles}), the formula
\begin{equation}
    d = \min(\mathrm{csys}_1[X(C_1, C_2)], \mathrm{csys}_1^*[X(C_1, C_2)]),
\end{equation}
is shown. The lower bound from the classical distances follows from relaxing the minimization:
\be
d  \geq \min\left[\min_{\mathbf a \in C_1 } \mathrm{wt}(\mathbf a),\min_{\mathbf b \in C_2} \mathrm{wt}(\mathbf b) \right]\;,
\ee
corresponding to $d \geq \min(d_1, d_2)$.
\end{proof}

The case for $\mathbf P_1 = (X_1, X_2, \cdots, X_n)$ and $\mathbf P_2 = (Z_1, Z_1, \cdots, Z_n)$ are usually   referred to as  CSS codes \cite{calderbank1996good,steane1996multiple}. Motivated by this, we say that a  stabilizer code has a \emph{CSS structure} if there  exists $(\mathbf P_1, C_1, \mathbf P_2, C_2)$ such that $\mathcal S= \mathcal S_{\mathbf P_1, \mathbf P_2}(C_1,C_2)$. It is the main result of this section to show that (local) projected Majorana codes are of this type. Consider the color decomposition of the faces: $F_\sigma =  F^R_\sigma \cup F^B_\sigma$. Then, we define the following chain complex:
\be
X(C_R,C_B): \mathrm{span}_{\mathbb F_2} F^R_\sigma \xrightarrow{h^T_R}\mathbb F_2^n \xrightarrow{h_B} \mathrm{span}_{\mathbb F_2} F^B_\sigma\;, 
\ee
where $h_R$ ($h_B$) maps vertices to the red (blue) faces it is tangent to. We claim:

\noindent
\begin{proposition}
Let $(\mathcal M, \Sigma)$ be a two-colorable, two dimensional Majorana fermion group. Then, 
it exists an anticommuting pair of Pauli tuples $\mathbf P_R, \mathbf P_B \in \mathcal P_n^{\times n}$ such that its local projection has a CSS structure, 
$P_\partial \mathcal M = \mathcal S_{\mathbf P_R, \mathbf P_B}(C_R, C_B)$, with classical codes $C_R = \ker(h_R) \subseteq \mathbb F_2^n$ and $C_B = \ker(h_B) \subseteq \mathbb F_2^n$ being the binary vector spaces spanned by the support of red and blue faces, respectively.
\end{proposition}
\begin{proof}
\noindent
By definition, $P_\partial \mathcal M = \mathrm{im}_\mathrm P[\mathrm C^\partial_\mathcal M(\mathcal D)]$. We will reconstruct the anticommuting Pauli tuples via the topological constraints of the isomorphism $\mathrm C_\mathcal M(\mathcal D) \cong B_1(X_\sigma)$. 

Consider a vertex $v \in V$. Since $G_\mathcal M$ has degree 4 on all vertices, there must be four edges incident to it, labeled as $\{e_{0,v}, e_{1,v}, e_{2,v}, e_{3,v}\}$, where we assumed orientability of $\Sigma$  to consistently place the edges in a clockwise orientation. We assume, without loss of generality, that the corresponding Majorana bilinears on each edge have support on $E_v \equiv \{\gamma^0_v, \gamma^1_v, \gamma^2_v, \gamma^3_v\}$, respectively.
Due to the two-colorability of the embedding, there will be two red faces and two blue faces, denoted $\{f_{v,R,1}, f_{v,R,2}, f_{v,B,1}, f_{v,B,2}\}$, as illustrated in Fig. \ref{fig:vertex}(a).

\begin{figure}
    \centering
    \includegraphics[width=0.9\linewidth]{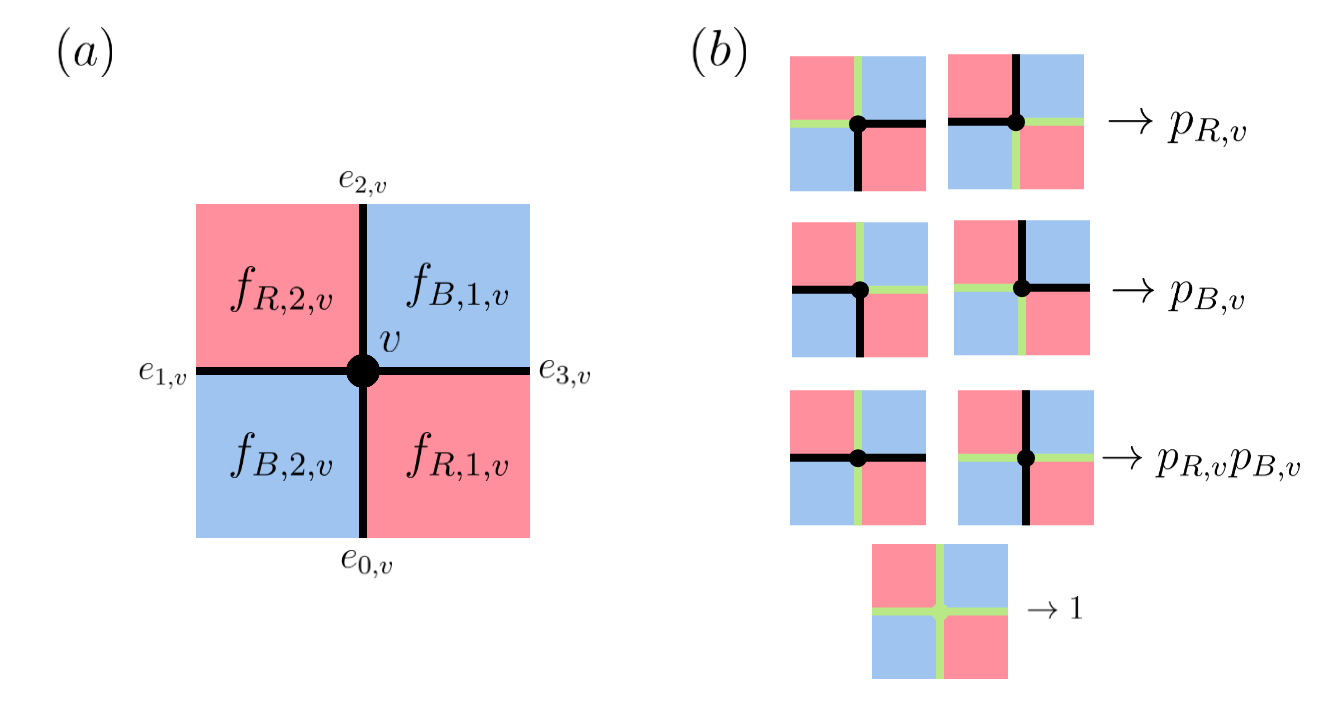}
    \caption{(a): Overall structure of a vertex $v \in V$ of a projected graph of a two-colorable $\mathcal M$. Each vertex has degree four, having four faces, two of each color, around it. (b) All seven options of $m_c$ in terms of their support on some $v$. Elements of $\mathrm{supp}_v(m_c)$ are highlighted in green, and the corresponding images after the   parton map are defined (modulo phases).}
    \label{fig:vertex}
\end{figure}

Consider now $m_c \in \mathrm{C}_\mathcal M(\mathcal D)$, and consider a vertex $v \in V$ where $m_c$ has nontrivial support. Denote $\mathrm{supp}_v(m_c)\subseteq E$ as the edges which have support on $v$. Note that $m_c$ must correspond to either one of the six configurations at the top of Fig. \ref{fig:vertex} (b) or $\mathrm{supp}_v(m_c) = E_v$, illustrated at the bottom. We will not worry about this case, since it corresponds to a local parity, and thus generates a trivial Pauli operator after projection.

By the  parton map    in Table \ref{eq:kitaevpartontable}, the six possibilities of placing the edges with degree 2 split into three classes, depending on the color of the faces the edge touches. If the edge is tangent to red faces, then it generates a Pauli that we will denote $p_{R,v}$, and if it is tangent only to blue faces, then the corresponding Pauli is $p_{B,v}$. In the last case, where the edges touch both of them, their product $p_{R,v}p_{B,v}$ is the result of the   parton map. Note that, for   distinct Paulis,
\be
p_{R,v}p_{B,v} = -p_{B,v}p_{R,v}\;.
\ee
We can then define $\mathbf P_R  = (P_{R,1}, P_{R,2}, \cdots, P_{R,n})$ and $\mathbf P_B = (P_{B,1}, P_{B,2}, \cdots, P_{B,n})$ as:
\be
P_{R,v} = s_v p_{R,v} \quad ; \quad P_{B,v} = s^\prime p_{B,v}\;,
\ee
where $s_v, s^\prime_v \in \{\pm 1\}$ are signs fixed by the equation:
\be
\mathrm P(m_c) = \mathbf P_R^{a}\mathbf P_B^{b}\;,
\ee
where $a \in C_R^\perp = \mathrm{im}(h_R^T)$, and $b \in C_B^\perp = \mathrm{im}(h_B^T)$. Since there are $2n$ signs, and $\leq n$ generators of $P_\partial \mathcal M$, they can always be fixed and the decomposition satisfied. 
 \end{proof}

Let us now prove each of the properties stated in Theorem \ref{theorem1} separately. The results follow from Propositions 2 and 3.
 
1. Topological encoding. This can be proven by computing invariants of the chain complexes involved. Given any chain complex $X$ with chains $\{X_k\}$,  its Euler characteristic is defined as
\be
\chi(X) = \sum_k (-1)^k \mathrm{dim}X_k = \sum_k (-1)^k \mathrm{dim}H_k(X)\;,
\ee
where the   last equality is a standard result in algebraic topology \cite{hatcher2001algebraic}, proving that $\chi$ is a bona-fide topological invariant. From the definition, we compute it for $X_\sigma$ and $X(C_R, C_B)$:
\begin{align}
    &\chi(X_\sigma) = |F_\sigma|-|E|+|V| = (r_R + r_B)-|E|+n\;,\\
    &\chi[X(C_R,C_B)] = r_R-n+r_B\;.
\end{align}
Since $|E| = 2n$, it follows that $\chi(X_\sigma) = \chi[X(C_R, C_B)]$. Since, by Proposition 3 and 2, $k=\mathrm{dim}H_1[X(C_R, C_B)]$ and $H_\cdot(X_\sigma) \cong H_\cdot(\Sigma, \mathbb F_2)$, we can show the result by proving that the second and zeroth homology groups of both complexes are isormophic, by   
the Betti number decomposition of the Euler invariants.

For $X_\sigma$, it turns out that the two homology groups are already fixed by the two conditions imposed on $\Sigma$ \cite{hatcher2001algebraic}:
\begin{align}
    &\mathrm{\Sigma\;is\;orientable}  \Rightarrow  H_2(X_\sigma)  \cong H_2(\Sigma, \mathbb F_2) \cong \mathbb F_2\;,\\
    &\mathrm{\Sigma\;is\;simply\;connected}  \Rightarrow H_0(X_\sigma) \cong H_0(\Sigma, \mathbb F_2) \cong \mathbb F_2\;.
\end{align}
Furthermore, we can write down the corresponding nontrivial representatives:
\begin{align}
    & \left[\sum_{f \in F_\sigma}f\right] \in H_2(X_\sigma)\;,\\
    & \left[v\right] \in H_0(X_\sigma)\;,
\end{align}
where, in the second line, we can take any $v \in V = \{1,2, \cdots, n\}$. The interpretation is as follows: As $H_2(X_\sigma) \cong \ker(\partial_2)$, the only nontrivial combination of the faces that
vanishes under the boundary map is the linear superposition of all of them, which corresponds to an orientation for $\Sigma$. For the zeroth homology group, $H_0(X_\sigma) \cong X_0/ \mathrm{im}\partial_1$ measures what elements of $X_0$ can be connected by elements of $X_1$. Since $\Sigma$ is simply connected (and also  $G_\mathcal M$, which must be for the cellular embedding on $\sigma$ to be valid), any pair of vertices $v_1, v_2 \in V$ can be connected by some $x \in X_1$, such that $\partial_1 x = v_1 + v_2$. Hence, any vertex is a good representative for $H_0$.

In the following, we will compute the same homologies for $X(C_R,C_B)$. By the triviality of the inclusion $0 \hookrightarrow \mathbb F_2^{r_B}$, we also have: $H_2[X(C_R,C_B)] \cong \ker(h_B^T: \mathbb F_2^{r_B}\to \mathbb F_2^n)$. Again, consider a set $S$ of blue faces such that
\begin{equation}
    h_B^T\left(\sum_{f \in S} f\right) = 0\;.
\end{equation}
Solutions of this equation provide nontrivial representatives of homology. As illustrated in Fig. \ref{fig:vertex} (a) and discussed in the proof of Proposition 3, for each vertex $v$ there are two blue faces. If $S \neq \varnothing$, then it must correspond to the set of all faces  to cancel the vertices pairwise. Thus, $H_2[X(C_R, C_B)] \cong \mathbb F_2$.

Now, we want to show that any pair of faces $f_R, f_R^\prime$ can be obtained as the image of the parity check acting on a set of vertices. Imagine that they do not share a vertex (otherwise, we would be done, since such a vertex $v$ satisfies $h_R (v) = f_R + f^\prime_R$). Then, we denote by $v_R$ and $v_R^\prime$  a pair of vertices that are tangent to them, respectively. Then, since the projected graph is simply connected, there exists $l \in X_1$ such that $\partial l = v_R + v^\prime_R$. We assume, without loss of generality, that in the set of vertices along $l$, $V(l)$, only $v_R$ and $v^\prime_R$ are tangent to the end faces. Now, let $F_R(l)$ be the set of red faces whose boundary is intersected by  $l$. Consider
\be
V_R(l) \equiv \{v \in V(l) \;|\;\{f_{v,R,1}, f_{r,R,2}\}\subseteq F_R(l)\}\;,
\ee
that is, the set of vertices whose corresponding tangent red faces lie in the path. Then, it follows that
\be
h_R\left(v_R + \sum_{v \in V_R(l)} v+ v^\prime_R\right) = f_R + f^\prime_R\;,
\ee
since only the red faces at the endpoints do not cancel upon the action of the parity check. Hence, $H_0[X(C_R, C_B)]\cong \mathbb F_2$.

2. String-like symmetries. This result follows from the discussion of Lemma 2 and Proposition 1. Let $m \in \mathrm{C}_\mathcal M(\mathcal D)$ be a Majorana operator which is mapped into $S(m)\in Z_1(X_\sigma)\backslash B_1(X_\sigma)$ in the isomorphism $S:\mathrm C_\mathcal M(\mathcal D) \to Z_1(X_\sigma)$. By definition, $\mathrm P(m) \in P_h \mathcal M$. Since every vertex in the support of $S(m)$ must have degree at least two, and the support of $\mathrm P (m)$ is upper bounded by the number of vertices in the support of $m$, we have the inequality:
\be
|\mathrm{supp}\mathrm P(S^{-1}(z))| \leq \mathrm{wt}S(z) \quad ; \quad \forall z \in Z_1(X_\sigma) \backslash B_1(X_\sigma)\;,
\ee
where we denote $z = S^{-1}(m)$. Minimizing over $z$ yields the combinatorial 1-systole, giving our desired result: $d \leq \mathrm{csys}_1(X_\sigma)$, which is indeed the minimal length cycle in $G_ \mathcal M$.

3. Anyon excitations. Although our construction does not involve a Hamiltonian for the spin state,   we can define   anyon excitations by  classifying the action of local operators on $|P\psi_\mathcal M\rangle$. We will do so  by first introducing Wilson line operators for the anyons  and then showing that any nontrivial Pauli operator can be decomposed in terms of these elementary excitations.

Let us recall the definition of an Abelian anyon theory \cite{wang2020and}, taking the pragmatic approach of Ref. \cite{ellison2022pauli}, which also discusses how to construct anyon data from stabilizer models. It is defined by a tuple $\mathcal A = (A, \theta)$, where $A$ is an Abelian group and $\theta: A  \to U(1)$ satisfying
\begin{itemize}
    \item $\theta(a^n) = \theta(a)^n$, $\forall a \in A$, $n \in \mathbb Z_{>0}$, 
    \item The function $B_\theta: A \times A \to U(1)$, defined by $B_\theta(a, a^\prime) 
    \equiv \theta(a a^\prime)/[\theta(a)\theta(a^\prime)]$, satisfying $B_\theta(a^n, a^\prime) =B_\theta(a,[a^\prime]^n)=B_\theta(a,a^\prime)$.
\end{itemize}
The function $\theta$ assigns a phase $\theta(a) \equiv e^{2\pi i q(a)}$ to the anyon $a$, obtained when two identical anyons are exchanged.

For the case of the toric code, we denote $\mathcal T \mathcal C = (A_\mathrm{TC}, \theta)$, with $A_\mathrm{TC}  = \{1,e,m,f\}$, endowed with the product $\times$, satisfying
\begin{align}
    1 &= e\times e   = m \times m = f\times f\;,\\
    f &= e\times m\;.
\end{align}
Noting that the $e$ and $m$ anyons generate the group, we can write $A_\mathrm{TC} \cong \langle e, m \rangle$. The corresponding phases are
\begin{equation}
\theta(1) = \theta(e) = \theta(m) =-\theta(f) = 1\;.
\end{equation}
To construct directly from the anyon data from the  microscopic lattice model, we will follow the prescription of Refs. \cite{levin2006quantum,kawagoe2020microscopic,haah2023nontrivial}. Let $\gamma_R, \gamma_B \in \mathbb F_2^n$ be elements such that $h_R(\gamma_R) = f_{R,1} + f_{R,2}$ and $h_B(\gamma_B) = f_{B,1} + f_{B,2}$. Then, we define
\begin{equation}
    W^e_{\gamma_R} \equiv \mathbf P_B^{\gamma_R }\quad ; \quad W^m_{\gamma_B} \equiv \mathbf P_R^{\gamma_B}\;.
    \label{eq:wilsoneandm}
\end{equation}
Note that they commute with all stabilizer generators, with the exception of the ones coming from the parity-check images. Indeed, we claim that they correspond to microscopic realization of Wilson line operators for the $e$ and $m$ anyons, by creating a pair of anyons at the corresponding endpoints. Their topological nature is immediate since their charges are invariant upon taking $(W_{\gamma_R}^e, W_{\gamma_B}^m) \to (W_{\gamma^\prime_R}^e, W_{\gamma^\prime_B}^m)$, with $\gamma^\prime_R =\gamma_R + z_R$ and $\gamma^\prime_B = \gamma_B + z_B$, where $z_R \in \mathrm{ker}(h_R)$ and $z_B \in \mathrm{ker}(h_B)$. We  denote the group generated by all such Wilson lines by $\mathcal W \equiv \langle \{W^e_{\gamma_R}, W^m_{\gamma_B} \}_{\gamma_R, \gamma_B} \rangle$.

Now, consider any $O \in \mathcal{P}_n$, which must have the following decomposition:
\be
O \propto \mathbf P_R^a \mathbf P_B^b\;,
\ee
with some $a,b \in \mathbb F_2$. Denote the $F_R(O), F_B(O) \subseteq F_\sigma$ as the set of faces with which  $O$ anticommutes. These correspond to elements $f_R \in X_2$, $f_B \in X_2$ whose parity-check images satisfy $a \cdot h_B(f_B) = 1\mod 2$ and $b \cdot h_R(f_R) =1\mod2$. We have the constraint
\begin{align}
&\sum_{f_B \in F^B_\sigma} a \cdot h_B(f_B)=0\mod2 \;,\\
&\sum_{f_R \in F^R_\sigma} b \cdot h_R(f_R) = 0\mod 2\;,
\end{align}
since $\sum_{f_B}h_B(f_B) = \sum_{f_R} h_R(f_R) = 0$, due to the fact that every edge is tangent to two red and blue faces. It follows  that $|F_R(O)|$ and $|F_B(O)|$ are both even. But then, by the fact that $H_0[X(C_R, C_B)] \cong H^0[X(C_R, C_B)]\cong \mathbb F_2$, the red and blue faces can be paired up, and thus $a$ $b$ can be decomposed into $a = \sum_l \gamma_{R, l}$ and $b= \sum_n \gamma_{B,n}$ such that
\begin{align}
&h_B \left(\sum_{l=1}^{|F_B(O)|/2} \gamma_{R,l}\right) = \sum_{f_B \in F_B(O)}f_B\\
&h_R\left(\sum_{n=1}^{|F_R(O)|/2}\gamma_{B,n}\right) =\sum_{f_B \in F_B(O)}f_B\;,
\end{align}
meaning that every operator can be written as
\be
O \propto \prod_l W^e_{\gamma_{R,l}} \prod_{n} W^m_{\gamma_{B,n}} \in \mathcal W\;.
\ee
This proves the isomorphism $\mathcal W \cong \mathcal P_n$, showing that every Pauli can be written as a product of Wilson lines. We note that there is a homomorphism $\mathcal W \to A_\mathrm{TC}$ defined on the generators as
\be
W^e_{\gamma_R} \mapsto e \quad ; \quad W^m_{\gamma_B} \mapsto m\;.
\ee
The $f$ anyon is obtained under this homomorphism by the image of $W^f_{\gamma_R, \gamma_B}\equiv W^e_{\gamma_R} W^m_{\gamma_B}$.
Hence, we have proven that every operator action on the code  corresponds to a toric code superselection sector.
 
The self-statistics of the anyon is also encoded in the algebra of Wilson lines. Let $\gamma_1, \gamma_2, \gamma_3$ be three anyon paths oriented clockwise around $v \in V$. Then, given $a \in A$ with Wilson line supported on $\gamma$ as $W^a_\gamma$, $\theta(a)$ generically satisfies
\begin{equation}
W^a_{\gamma_1}[W^a_{\gamma_2}]^\dagger W^a_{\gamma_3} = \theta(a )W^a_{\gamma_3}[W^a_{\gamma_2}]^\dagger W^a_{\gamma_1}\;,
\end{equation}
For our case of Hermitian Wilson lines, we have $W^a_{\gamma_1}W^a_{\gamma_2} W^a_{\gamma_3} = \theta(a )W^a_{\gamma_3}W^a_{\gamma_2} W^a_{\gamma_1}$. For the $e$ and $m$ anyons, all Wilson lines commute among themselves, made clear by the definition in Eq. (\ref{eq:wilsoneandm}). Thus, $\theta(e)  = \theta(m)$. However, two lines  for the $f$ anyon anticommute if they share a vertex. Hence, exchanging the order gives us $\theta(f)=-1$. This shows that the full anyon theory does indeed correspond to $\mathcal T \mathcal C$ and thus proves the theorem.  

\bibliography{refs}

\end{document}